\documentclass[aps,pra,10pt,amsmath,twocolumn,floatfix,eqsecnum,superscriptaddress,longbibliography]{revtex4-1}
\usepackage[dvipsnames]{xcolor}
\usepackage[colorlinks=true,bookmarks=false,linkcolor=NavyBlue,urlcolor=NavyBlue,citecolor=NavyBlue,breaklinks]{hyperref}
\usepackage{amsmath}
\usepackage{amsfonts}
\usepackage{tabularx}
\usepackage{graphicx}
\usepackage{times}
\usepackage{mathtools}
\usepackage{braket}
\usepackage{amsthm}
 
\newtheorem{theorem}{Theorem}
\newtheorem{corollary}{Corollary}
\newtheorem{proposition}{Proposition}
\newtheorem{lemma}{Lemma}
\theoremstyle{remark}

\newcommand{\parti}[2]{\frac{\partial #1}{\partial #2}}

\newcommand{\avg}[1]{\langle#1\rangle}
\newcommand{\Avg}[1]{\left\langle#1\right\rangle}
\newcommand{\abs}[1]{\left|#1\right|}
\newcommand{\bk}[1]{\left(#1\right)}
\newcommand{\Bk}[1]{\left[#1\right]}
\newcommand{\BK}[1]{\left\{#1\right\}}

\newcommand{\iverson}[1]{1_{#1}}
\newcommand{\norm}[1]{\lVert #1 \rVert}
\newcommand{\opnorm}[1]{\norm{#1}_{\rm op}}

\DeclareMathOperator*{\argmin}{arg\,min}
\DeclareMathOperator{\spn}{span}
\DeclareMathOperator{\cspn}{\overline{span}}
\DeclareMathOperator{\trace}{tr}
\DeclareMathOperator{\real}{Re}
\DeclareMathOperator{\imag}{Im}

\DeclareMathOperator{\supp}{supp}

\newcommand{\PRXagain}[1]{#1}

\begin{document}

\title{Quantum Semiparametric Estimation}

\author{Mankei Tsang}
\email{mankei@nus.edu.sg}
\homepage{https://blog.nus.edu.sg/mankei/}
\affiliation{Department of Electrical and Computer Engineering,
  National University of Singapore, 4 Engineering Drive 3, Singapore
  117583}

\affiliation{Department of Physics, National University of Singapore,
  2 Science Drive 3, Singapore 117551}

\author{Francesco Albarelli}
\email{francesco.albarelli@gmail.com}
\affiliation{Faculty of Physics, University of Warsaw, 02-093 Warszawa, Poland}
\affiliation{Department of Physics, University of Warwick, Coventry CV4 7AL, United Kingdom}

\author{Animesh Datta}
\email{animesh.datta@warwick.ac.uk}
\affiliation{Department of Physics, University of Warwick, Coventry CV4 7AL, United Kingdom}

\date{\today}

\begin{abstract}
  In the study of quantum limits to parameter estimation, the high
  dimensionality of the density operator and that of the unknown
  parameters have long been two of the most difficult challenges.
  Here we propose a theory of quantum semiparametric estimation that
  can circumvent both challenges and produce simple analytic bounds
  for a class of problems \PRXagain{in which the dimensions are
    arbitrarily high, few prior assumptions about the density operator
    are made, but only a finite number of the unknown parameters are
    of interest.}  \PRXagain{We also relate our bounds to Holevo's
    version of the quantum Cram\'er-Rao bound, so that they can
    inherit the asymptotic attainability of the latter in many cases
    of interest.}  The theory is especially relevant to the estimation
  of a parameter that can be expressed as a function of the density
  operator, such as the expectation value of an observable, the
  fidelity to a pure state, the purity, or the von Neumann entropy.
  Potential applications include quantum state characterization for
  many-body systems, optical imaging, and interferometry, where full
  tomography of the quantum state is often infeasible and only a few
  select properties of the system are of interest.
\end{abstract}

\maketitle
\section{Introduction}

The random nature of quantum mechanics has practical implications for
the noise in sensing, imaging, and quantum-information applications
\cite{helstrom,demkowicz15,paris,*glm2011,*szczykulska16,*pirandola18,*braun18,*pezze18,*albarelli20a,tnl,tsang19a,kolobov,*kolobov07,*kolobov_fabre,*jezek,*hradil,*taylor16,*genovese16,*berchera19,*moreau19}.
To derive their fundamental quantum limits, one standard approach is
to compute quantum versions of the Cram\'er-Rao bound
\cite{helstrom,demkowicz15,paris,tnl,tsang19a,holevo11,hayashi,hayashi05}. In
addition to serving as rigorous limits to parameter estimation, the
quantum bounds have inspired new sensing and imaging paradigms that go
beyond conventional methods \cite{paris,tnl,tsang19a}.

The study of quantum limits has grown into an active research field
called quantum metrology in recent years, building on the pioneering
work of Helstrom \cite{helstrom} and Holevo \cite{holevo11}. A major
current challenge is the computation of quantum bounds for
high-dimensional density operators and high-dimensional parameters, as
the brute-force method quickly becomes intractable for increasing
dimensions; see
Refs.~\cite{yuan17,*genoni19,*chabuda20,*fanizza19,albarelli19} for a
sample of recent efforts to combat the so-called curse of
dimensionality. Most of the existing methods, however, ultimately have
to resort to numerics for high dimensions. While numerical methods are
no doubt valuable, analytic solutions should be prized higher---as
with any study in physics---for their simplicity and offer of
insights. Unfortunately, except for a few cases where one can exploit
the special structures of the density-operator family
\cite{helstrom,holevo11,twc,*guta11,ng16,zhou19,tsang19}, analytic
results for high-dimensional problems remain rare in quantum
metrology.

Here we propose a theory of quantum semiparametric estimation that can
turn the problem on its head and deal with density operators with
arbitrarily high dimensions and little assumed structure. The theory
is especially relevant to the estimation of a parameter that can be
expressed as a function of the density operator, such as the
expectation value of an observable, the fidelity to a given pure
state, the purity, or the von Neumann entropy. The density operator is
assumed to come from an enormous family, its dimension can be
arbitrarily high and possibly infinite, and the unknown ``nuisance''
parameters have a similar dimension to that of the density operator.
Despite the seemingly bleak situation, our theory can yield
surprisingly simple analytic results, precisely because of the absence
of structure. Our results are ideally suited to scientific
applications, such as quantum state characterization
\cite{guehne09,paris04,*lvovsky09,*horodecki,*filip02,*brun04,*flammia11,*enk12,horodecki03}
optical imaging \cite{helstrom,tsang19a,kolobov,zhou19,tsang19}, and
interferometry \cite{helstrom,holevo11,demkowicz15,paris}, where the
dimensions can be high, the density operator is difficult to specify
fully, and it is prudent to assume little prior information.

The theory set forth generalizes the deep and exquisite theory of
semiparametric estimation in classical statistics
\cite{ibragimov81,bickel93,tsiatis06}, which has seen wide
applications in fields such as biostatistics \cite{tsiatis06},
econometrics \cite{newey90}, astrostatistics \cite{feigelson12}, and,
most recently, optical superresolution \cite{tsang19b}.  By necessity,
the classical theory involves infinite-dimensional spaces for random
variables and makes extensive use of geometric and Hilbert-space
concepts. As will be seen later, the operator Hilbert space introduced
by Holevo \cite{holevo11,holevo77} turns out to be the right arena for
the quantum case, and the geometric picture of quantum states
\cite{hayashi,hayashi05,amari,uhlmann93,*braunstein,*bengtsson,*sidhu20}
can provide illuminating insights.

\PRXagain{Our formalism is primarily based on Helstrom's version of
  the quantum Cram\'er-Rao bound \cite{helstrom}. While this allows us
  to adapt the classical methods more easily, it is unable to account
  for the increased errors due to the incompatibility of quantum
  observables when multiple parameters are involved
  \cite{holevo11,demkowicz20}.  We address this issue by studying also
  Holevo's version of the quantum Cram\'er-Rao bound \cite{holevo11}
  in the semiparametric setting and proving that the two versions turn
  out to be close. This result enables our bounds to inherit the
  asymptotic attainability of Holevo's bound
  \cite{kahn09,gill_guta,*yamagata13,*yang19,demkowicz20} in many
  cases of interest.}

\section{\label{sec_preview}Preview of typical results}
Before going into the formalism, we present some typical results of the
theory to offer motivation.

Suppose that an experimenter has received $N$ quantum objects, such as
atoms, electrons, photons, or optical pulses, each with the same
quantum state $\rho$.  The experimenter would like to estimate a
parameter $\beta$ as a function of $\rho$.
Without any knowledge or assumption about $\rho$, what is the best
measurement to perform for the estimation of $\beta$, and what is the
fundamental limit to the precision for any measurement?

The quantum semiparametric theory can provide simple answers to the
above questions.  For the simplest example, let
$\beta = \trace \rho Y$, where $Y$ is a given observable, and assume
that the estimator is required to be unbiased. For example, one
may wish to estimate
\begin{enumerate}
\item the mean position of photons or electrons in optical or electron
  microscopy,

\item the mean photon number in an optical mode in optical sensing,
  imaging, and communication \cite{helstrom},

\item the mean energy, momentum, or field of quantum particles in
  particle-physics, condensed-matter, or quantum-chemistry
  experiments,

\item a density-matrix element, the fidelity
  $\bra{\psi}\rho\ket{\psi}$ to a target pure state $\ket{\psi}$, or
  an entanglement witness in quantum-information experiments
  \cite{guehne09,paris04}.
\end{enumerate}
This problem appears in all areas of quantum mechanics
\cite{schwartz,*chaikin,*haken,*bravyi19}, as most quantum
calculations offer predictions in terms of expectation values only,
and experiments that aim to estimate the expectation values and verify
the predictions with few assumptions about the density operator are in
essence semiparametric estimation.  The theory here shows that the
optimal measurement is simply a von Neumann measurement of the
observable $Y$ of each copy of the objects, followed by an average of
the outcomes.  For any measurement, the mean-square error of the
estimation, denoted by the sans-serif $\mathsf{E}$, has a quantum
limit given by
\begin{align}
\mathsf{E} &\ge \frac{1}{N}\trace \rho (Y-\beta)^2.
\label{MSE_Y}
\end{align}
Absent any information about $\rho$, the separate measurements and the
sample mean seem to be the most obvious procedure, but it is not at
all obvious that it is optimal, given the infinite possibilities
allowed by quantum mechanics.

While Eq.~(\ref{MSE_Y}) has been derived before via a more
conventional method for a finite-dimensional $\rho$
\cite{watanabe10,*watanabe11}, our theory can also deal with
infinite dimensions as well as more advanced examples in quantum
information and quantum thermodynamics.  For example, if the parameter
of interest is the purity $\beta = \trace \rho^2$, the bound is
\begin{align}
\mathsf{E} &\ge \frac{4}{N} \trace \rho (\rho-\beta)^2,
\label{MSE_purity}
\end{align}
and if the parameter is the relative entropy
$\beta = \trace \rho(\ln\rho-\ln\sigma)$ with respect to a target
state $\sigma$, the bound is
\begin{align}
\mathsf{E} &\ge \frac{1}{N} \trace \rho (\ln\rho-\ln\sigma-\beta)^2.
\label{MSE_entropy}
\end{align}
For these two examples, the bounds are asymptotically attainable in
principle, at least when $\rho$ is finite-dimensional
\cite{kahn09,gill_guta,*yamagata13,*yang19,demkowicz20}.

The semiparametric theory is relevant to experiments on many-body
quantum systems and quantum simulation \cite{bloch,*georgescu14},
because often there is no simple model for $\rho$, full tomography of
$\rho$ is infeasible, and only a few select properties of the system
may be of interest. Although a significant literature in quantum
information has been devoted to such semiparametric problems
\cite{guehne09,paris04,*lvovsky09,*horodecki,*filip02,*brun04,*flammia11,*enk12,horodecki03},
their connections to the classical theory have not yet been
recognized. By generalizing the classical theory, this work
establishes fundamental limits to the task, indicating the minimum
amount of resources needed to achieve a desired precision and also
offering a rigorous yardstick for experimental design.  This work thus
addresses a foundational question by Horodecki \cite{horodecki03}:
\emph{``What kind of information (whatever it means) can be extracted
  from an unknown quantum state at a small measurement cost?''} Our
work shows that quantum metrology---and quantum semiparametric
estimation in particular---offers a viable attack on the question via
a statistical notion of efficiency.

An extension of the above scenario is the estimation of $\beta$ given
a constraint on $\rho$. For example, suppose that the quantum state is
known to possess a mean energy $\trace \rho H = E$, where $H$ is the
Hamiltonian, or attain a fidelity of $\bra{\phi}\rho\ket{\phi} = F$
with respect to another pure state $\ket{\phi}$. How may this new
information affect the estimation? Write the constraint as
$\trace \rho Z = \zeta$, where $Z$ is an observable and $\zeta$ is a
given constant. The quantum bound for the $\beta = \trace \rho Y$
example turns out to be
\begin{align}
\mathsf{E} &\ge \frac{1}{N}\bk{V_Y- \frac{C_{YZ}^2}{V_Z}},
&
V_Y &= \trace \rho (Y-\beta)^2,
\label{correlation_bound}
\\
C_{YZ} &= \trace \rho (Y-\beta)\circ (Z-\zeta),
&
V_Z &= \trace \rho (Z-\zeta)^2,
\label{correlation_bound2}
\end{align}
where $A \circ B = (AB+BA)/2$ denotes the Jordan product. The bound is
reduced by the correlation between $Y$ and $Z$. 

Another paradigmatic problem in quantum metrology is displacement
estimation \cite{helstrom,holevo11,demkowicz15,paris}, which can be
modeled by
\begin{align}
\rho = \exp\bk{-iH\beta}\rho_0 \exp\bk{iH\beta},
\label{displacement}
\end{align}
where $\rho_0$ is the initial state, $H$ is a generator, such as the
photon-number operator in optical interferometry, and $\beta$ is the
displacement parameter to be estimated.  Applications range from
optical and atomic interferometry to atomic clocks, magnetometry,
laser ranging, and localization microscopy
\cite{demkowicz15,paris,kolobov}. If nothing is known about $\rho_0$
other than a constraint $\trace \rho_0 Z = 0$, the quantum bound
turns out to be
\begin{align}
\mathsf{E} &\ge 
\frac{\trace \rho_0Z^2}{N\{-i\trace \rho_0 [Z,H]\}^2},
\label{displacement_bound}
\end{align}
where $[Z,H] \equiv ZH-HZ$. Our theory can in fact give similarly
simple results for a class of such semiparametric problems.

It must be stressed that, apart from the underlying Hilbert space and
the constraints discussed above, the experimenter is assumed to know
nothing about the density operator, and the bounds here are
valid regardless of its dimension. The existing method of deriving
such quantum limits is to model $\rho$ with 
many parameters \cite{hayashi,hayashi05,kahn09,watanabe10},
compute a quantum version of the Fisher information matrix, and then
invert it. This brute-force method is rarely feasible for problems
with high or infinite dimensions. A new philosophy is needed.

In the next sections, we present the theory of quantum semiparametric
estimation in increasing sophistication. Sections~\ref{sec_direct} and
\ref{sec_submodels} generalize the quantum Cram\'er-Rao bound proposed
by Helstrom \cite{helstrom} in a geometric picture. While the picture
is not new \cite{hayashi05,amari}, it has so far remained an
intellectual curiosity only.  Sections~\ref{sec_direct} and
\ref{sec_submodels} show that it can in fact give simple solutions,
such as Eqs.~(\ref{MSE_Y})--(\ref{MSE_entropy}), to a class of
semiparametric problems with arbitrary
dimensions. Section~\ref{sec_direct} establishes the general formalism
and also proves results that are valid for finite dimensions, while
Sec.~\ref{sec_submodels} deals with the infinite-dimensional case via
an elegant concept called parametric submodels. In the classical
theory, the concept was first adumbrated by Charles Stein
\cite{stein56} and developed by Levit and many others
\cite{ibragimov81,bickel93,tsiatis06}. Section~\ref{sec_con} further
develops the formalism to account for constraints on the
density-operator family, in order to produce results such as
Eq.~(\ref{correlation_bound}). An example of entropy estimation in
quantum thermodynamics is also discussed there.
Section~\ref{sec_optics} discusses some practical problems in optics
and summarizes existing results on incoherent optical imaging
\cite{tsang19a} in the language of quantum semiparametrics, in order
to provide a more concrete context for the formalism.
Section~\ref{sec_displacement} considers semiparametric estimation in
the presence of explicit nuisance parameters and studies in particular
the problem of displacement estimation with a poorly characterized
initial state, in order to produce results such as
Eq.~(\ref{displacement_bound}). To complete the formalism,
Section~\ref{sec_multi} considers a vectoral parameter of interest and
Holevo's version of the quantum Cram\'er-Rao bound
\cite{holevo11}. \PRXagain{There we prove that the Helstrom and Holevo
  bounds are equal if the parameter of interest is a scalar, and they
  remain within a factor of two of each other in the vectoral
  case. The latter fact generalizes a recent result in the parametric
  setting \cite{carollo19,*carollo20}.  Thus the Helstrom version can
  inherit the asymptotic attainability of the latter
  \cite{kahn09,gill_guta,demkowicz20} to within a factor of two.}

\section{\label{sec_direct}Geometric picture of quantum estimation
  theory}
This section is organized as follows. Section~\ref{sec_helstrom}
introduces the Helstrom bound in the conventional formulation.
Section~\ref{sec_holevo} introduces some important Hilbert-space
concepts, including the tangent space and the influence operators.
Section~\ref{sec_GHB} generalizes the Helstrom bound in terms of a
projection of an influence operator into the tangent
space. Section~\ref{sec_gradient} shows how an influence operator can
be derived for a given parameter of interest, while
Sec.~\ref{sec_tangent} proves that the tangent space is simple if the
density operator is assumed to be finite-dimensional but otherwise
arbitrary.  The projection is then straightforward, and
Sec.~\ref{sec_tangent} demonstrates the derivation of
Eqs.~(\ref{MSE_Y})--(\ref{MSE_entropy}) as examples.

\subsection{\label{sec_helstrom}Helstrom bound}
Let 
\begin{align}
\mathbf F \equiv \BK{\rho(\theta): \theta \in 
\Theta \subseteq \mathbb R^p}
  \label{family}
\end{align}
be a family of density operators parametrized by
$\theta = (\theta_1,\dots,\theta_p)^\top$, where the superscript
$\top$ denotes the matrix transpose and $p$ denotes the dimension of
the parameter space $\Theta$. The operators are assumed to operate on
a common Hilbert space $\mathcal H$, with an orthonormal basis
\begin{align}
\BK{\ket{j}: j  \in \mathcal Q, \braket{j|k} = \delta_{jk}}
\label{Q}
\end{align}
that does not depend on $\theta$. Let 
\begin{align}
d &\equiv \dim \mathcal H = \abs{\mathcal Q}
\end{align}
be the dimension of $\mathcal H$, which may be infinite.  The family
is assumed to be smooth enough so that any
$\partial_j \equiv \partial/\partial\theta_j$ can be interchanged with
the operator trace $\trace$ in any operation on $\rho(\theta)$.
Define $\partial \equiv (\partial_1,\dots,\partial_p)^\top$,
  and define a vector of operators $S \equiv (S_1,\dots,S_p)^\top$ as
solutions to
\begin{align}
  \partial \rho &= \rho \circ S,
\label{sld}
\end{align}
which is shorthand for the system of equations
\begin{align}
\left.\partial_j\rho(\theta)\right|_{\theta=\phi}
= \rho(\phi) \circ S_j(\phi),
\quad j = 1,\dots,p.
\end{align}
$\phi$ is the true parameter value, and all functions of $\theta$ in
this section are assumed to be evaluated implicitly at the same
$\theta = \phi$. Each $S_j$ is called a symmetric logarithmic
derivative in the quantum metrology literature, but here we call it a
score, in accordance with the statistics terminology
\cite{ibragimov81,bickel93,tsiatis06}. All vectors are assumed to be
column vectors in this paper.

To model a measurement, define a positive operator-valued measure
(POVM) $E$ on a measurable space $(\mathcal X,\Sigma_{\mathcal X})$,
where $\Sigma_{\mathcal X}$ is the sigma algebra on the set
$\mathcal X$. Let the parameter of interest be a scalar
$\beta(\theta) \in \mathbb R$; generalization for a vectoral $\beta$
will be done in Sec.~\ref{sec_multi}.  Assume an estimator
  $\check\beta: \mathcal X \to \mathbb R$ that satisfies
\begin{align}
\int \check\beta(\lambda)\trace dE(\lambda)\rho  &= \beta,
&
\int \check\beta(\lambda)\trace dE(\lambda)\partial\rho  &= \partial\beta.
\label{unbiased}
\end{align}
$(E,\check\beta)$ is called a locally unbiased measurement, as we only
require Eqs.~(\ref{unbiased}) to hold at the true $\theta =
\phi$. Only local unbiasedness conditions are needed in this paper,
and for brevity we will no longer explicitly describe them as local.
Define the mean-square estimation error as
\begin{align}
\mathsf{E} &\equiv \int \Bk{\check\beta(\lambda)-\beta}^2 \trace dE(\lambda)\rho.
\end{align}
If $p < \infty$, a quantum version of the Cram\'er-Rao bound due to
Helstrom \cite{helstrom}, denoted by the sans-serif $\mathsf{H}$,
applies to any unbiased measurement and can be expressed as
\begin{align}
\mathsf{E} &\ge
\mathsf{H} \equiv (\partial\beta)^\top K^{-1} \partial\beta,
\label{HB}
\end{align}
where the Helstrom information matrix $K$ is defined as
\begin{align}
K_{jk} &\equiv \trace \rho\bk{S_j \circ S_k}.
\label{K}
\end{align}
The Helstrom bound sets a lower bound on the estimation error for any
quantum measurement and any unbiased estimator
\cite{helstrom,holevo11,hayashi,hayashi05}. The estimation of $\beta$
with an infinite-dimensional $\theta$ ($p = \infty$) is called
semiparametric estimation in statistics
\cite{ibragimov81,bickel93,tsiatis06}, although the methodology
applies to arbitrary dimensions. If $\theta$ is partitioned into
$(\beta,\eta_1,\eta_2,\dots)^\top$, then $\eta$ is called nuisance
parameters \cite{tsiatis06,suzuki19}.

\subsection{\label{sec_holevo}Hilbert spaces for operators}
We now follow Holevo \cite{holevo11,holevo77} and introduce operator
Hilbert spaces in order to generalize the Helstrom bound for
semiparametric estimation. The formalism may seem daunting at first
sight, but the payoff is substantial, as it simplifies proofs, treats
the infinite-dimensional case rigorously, and also enables one to
avoid the explicit computation of $S$ and $K^{-1}$ for a large class
of problems. In the following, we assume familiarity with the basic
theory of Hilbert spaces and the mathematical treatment of quantum
mechanics; see, for example,
Refs.~\cite{holevo11,debnath05,reed_simon}.

All operators considered in this paper are self-adjoint.  Consider
$\rho$ in the diagonal form
$\rho = \sum_j \lambda_j \ket{e_j}\bra{e_j}$ with $\lambda_j > 0$. The
support of $\rho$ is
$\supp(\rho) = \cspn\{\ket{e_j}\} \subseteq \mathcal H$, where $\cspn$
denotes the closed linear span. $\rho$ is called full rank if
$\supp(\rho) = \mathcal H$.  Define the weighted inner product between
two operators $h$ and $g$ as
\begin{align}
\Avg{h,g} &\equiv \trace \rho \bk{h \circ g},
\label{inner}
\end{align}
and a norm as
\begin{align}
\norm h &\equiv \sqrt{\Avg{h,h}},
\end{align}
not to be confused with the operator norm
$\opnorm h = \sup_{\ket{\psi} \in \mathcal H}
\sqrt{\bra{\psi}h^2\ket{\psi}} \ge \norm h$.  An operator is called
bounded if $\opnorm h < \infty$ and square summable with respect to
$\rho$ if $\norm h < \infty$, although all operators are bounded by
definition if $d < \infty$.  For two vectors of operators $A$ and $B$,
it is convenient to use $\avg{A,B}$ to denote a matrix with entries
\begin{align}
\Avg{A,B}_{jk} &= \Avg{A_j,B_k},
\label{inner_matrix}
\end{align}
such as $K = \avg{S,S}$ as a Gram matrix.

Define the real Hilbert space for square-summable operators
with respect to the true $\rho$ as \cite{holevo11,holevo77}
\begin{align}
\mathcal Y &\equiv \BK{h:\norm h < \infty}.
\label{Y}
\end{align}
To be precise, each Hilbert-space element is an equivalence
class of operators with zero distance between them, viz.,
$\{\hat h_j:\norm{\hat h_j-\hat h_k} = 0\ \forall j,k\}$.  The
distinction between an element and its operators is important only if
$\rho$ is not full rank; we put a hat on an operator if the
distinction is called for. Two important Hilbert-space elements are
the identity element $I$ and the zero element $0$; sometimes we will
substitute $I = 1$ for brevity.

Define a subspace of zero-mean operators as
\begin{align}
\mathcal Z &\equiv \BK{h \in \mathcal Y: \trace \rho h = \Avg{h,I}= 0},
\label{Z}
\end{align}
and the orthocomplement of $\mathcal Z$ in $\mathcal Y$ as
\begin{align}
\mathcal Z^\perp &\equiv \BK{h \in \mathcal Y:
\Avg{g,h} = 0\ \ \forall g \in \mathcal Z}
 = \spn\BK{I}.
\end{align}
In particular, the projection of any $h \in \mathcal Y$ into
$\mathcal Z^\perp$ is simply $\Pi(h|\mathcal Z^\perp) = \avg{h,I}$,
where $\Pi$ denotes the projection map, and
\begin{align}
\Pi(h|\mathcal Z) = h - \Pi(h|\mathcal Z^\perp) = h - \Avg{h,I}.
\end{align}
The most important Hilbert space in estimation theory is the tangent
space spanned by the set of scores
$\{S\} \equiv \{S_1,\dots,S_p\}$ 
\cite{ibragimov81,bickel93,tsiatis06}, generalized here as
\begin{align}
  \mathcal T &\equiv \cspn\BK{S} \subseteq \mathcal Z.
\label{tangent}
\end{align}
$\{S\}$ is also known as the tangent set. The condition
$\mathcal T \subseteq \mathcal Z$ requires the assumption
$K_{jj} = \avg{S_j,S_j} < \infty$ for all $j$; the zero-mean
requirement is satisfied because
$\avg{S,I} = \trace \partial \rho = \partial \trace\rho = 0$.  A
useful relation for any bounded operator $h$ is
\begin{align}
\Avg{S_j,h} &= \trace \rho (S_j \circ h) 
= \trace (\rho \circ S_j)  h = \trace (\partial_j \rho)h,
\label{cyclic}
\end{align}
via Ref.~\cite[Eq.~(2.8.88)]{holevo11}. Denote also the
orthocomplement of $\mathcal T$ in $\mathcal Z$ as
\begin{align}
\mathcal T^\perp &\equiv \BK{h\in \mathcal Z: \Avg{S,h} = 0},
\end{align}
which is useful if a projection of $h \in \mathcal Z$ into
$\mathcal T$ is desired and $\Pi(h|\mathcal T^\perp)$ is easier to
compute, since 
\begin{align}
\Pi(h|\mathcal T) = h- \Pi(h|\mathcal T^\perp).
\end{align}
Another important concept in the classical theory is the influence
functions \cite{ibragimov81,bickel93,tsiatis06}, which we generalize by defining
the set of influence operators as
\begin{align}
\mathcal D &\equiv
\BK{\delta \in \mathcal Z: \Avg{S,\delta} = \partial\beta}.
\label{influence_set}
\end{align}
These operators play a major role in Holevo's formulation of quantum
Cram\'er-Rao bounds \cite{holevo11,ragy16}, although their connection
to the classical concept did not seem to be appreciated before.

\subsection{\label{sec_GHB}Generalized Helstrom bound}
Let the error operator with respect to an unbiased measurement
be
\begin{align}
\delta = \int \check\beta(\lambda) dE(\lambda) - \beta.
\label{error}
\end{align}
It can be shown \cite[Sec.~6.2]{holevo11} that $\delta\in\mathcal D$
(as long as $\norm{\delta} < \infty$), and also that $\norm{\delta}^2$
bounds the estimation error as
\begin{align}
\mathsf{E} &  \ge \norm{\delta}^2.
\label{MSE}
\end{align}
A generalized Helstrom bound (GHB) for any unbiased
measurement, denoted by $\tilde{\mathsf{H}}$, can then be expressed as
\begin{align}
\mathsf{E} &\ge
\norm{\delta}^2 \ge 
\inf_{\delta \in \mathcal D} \norm{\delta}^2 \equiv \tilde{\mathsf{H}}.
\label{GHB}
\end{align}
We call an unbiased measurement efficient if it has an error that
achieves the GHB, following the common statistics terminology
\cite{ibragimov81,bickel93,tsiatis06}. 

Proofs that Eq.~(\ref{GHB}) is equal to Eq.~(\ref{HB}) if $p < \infty$
and $K^{-1}$ exists can be found in
Refs.~\cite{nagaoka89,amari,ragy16}.  The following theorem gives a
more general expression that is the cornerstone of quantum
semiparametric estimation.
\begin{theorem}
\label{thm_GHB}
\begin{align}
\tilde{\mathsf{H}} &= \min_{\delta \in \mathcal D} \norm{\delta}^2
= \norm{\delta_{\rm eff}}^2,
\label{GHB2}
\end{align}
where $\delta_{\rm eff}$, henceforth called the efficient influence,
is the unique element in the influence-operator set $\mathcal D$ given by
\begin{align}
\delta_{\rm eff} &=
\Pi(\delta|\mathcal T),
\label{influ_eff}
\end{align}
and $\Pi(\delta|\mathcal T)$ denotes the projection of any influence
operator $\delta \in \mathcal D$ into the tangent space $\mathcal T$.
\end{theorem}
\begin{proof}
  The proof is similar to the classical one
  \cite{bickel93,tsiatis06}. First note that, since
  $\mathcal D \subseteq \mathcal Z = \mathcal T \oplus \mathcal
  T^\perp$, any $\delta \in \mathcal D$ can always be decomposed into
\begin{align}
\delta &= \delta_{\rm eff} + h,
&
\delta_{\rm eff} &= \Pi(\delta|\mathcal T),
&
h &= \Pi(\delta|\mathcal T^\perp).
\end{align}
This implies
$\avg{S,\delta_{\rm eff}} = \avg{S,\delta-h} = \avg{S,\delta} =
\partial\beta$, and therefore $\delta_{\rm eff} \in \mathcal D$. Now
the Pythagorean theorem gives
\begin{align}
\norm{\delta}^2 &= \norm{\delta_{\rm eff}}^2 
+ \norm{h}^2 \ge \norm{\delta_{\rm eff}}^2,
\end{align}
which results in Eq.~(\ref{GHB2}).

To prove the uniqueness of $\delta_{\rm eff}$ in $\mathcal D$, suppose
that there exists another $\delta'\in\mathcal D$ that gives
$\norm{\delta'} = \norm{\delta_{\rm eff}}$.  Define
$g = \delta'-\delta_{\rm eff}$. Since
$\avg{S,g} = \avg{S,\delta'}-\avg{S,\delta_{\rm eff}} =
\partial\beta-\partial\beta = 0$, $g \in \mathcal T^\perp$, and the
Pythagorean theorem yields
$\norm{\delta'}^2 = \norm{\delta_{\rm eff}}^2+\norm{g}^2$.  This
implies that $\norm{g} = 0$ and $g = 0$, contradicting the assumption
that $\delta' \neq \delta_{\rm eff}$.  Hence $\delta_{\rm eff}$ must
be unique, and $\Pi(\delta|\mathcal T)$ for
any $\delta \in \mathcal D$ results in the same $\delta_{\rm eff}$.
\end{proof}

Figure~\ref{efficient_influence} illustrates all the Hilbert-space
concepts involved in Theorem~\ref{thm_GHB}. 

\begin{figure}[htbp!]
\centerline{\includegraphics[width=0.45\textwidth]{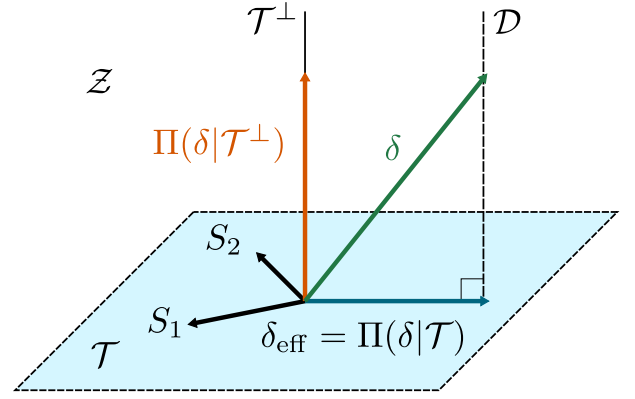}}
\caption{\label{efficient_influence} The whole space in the picture
  represents $\mathcal Z$, the space of zero-mean operators.
  $\mathcal T$ is the tangent space spanned by the tangent set
  $\{S\}$.  $\mathcal T^\perp$ is the orthocomplement, which contains
  elements orthogonal to all the scores. $\mathcal D$ is the set of
  influence operators, which all have a fixed projection in
  $\mathcal T$ determined by $\partial\beta$.  $\delta$ is an
  influence operator in $\mathcal D$. The projection of $\delta$ into
  $\mathcal T$ gives the efficient influence $\delta_{\rm eff}$, which
  has the smallest norm among all the influence operators.
  $\Pi(\delta|\mathcal T^\perp)$ is the projection of $\delta$ into
  $\mathcal T^\perp$.}
\end{figure}

Before we apply the theorem to examples, we list a couple of important
corollaries. The first corollary reproduces the original Helstrom
bound given by Eq.~(\ref{HB}) and is expected from earlier
derivations; see, for example, Ref.~\cite[Eq.~(20) in
Chap.~18]{hayashi05} and Ref.~\cite[Eq.~(7.93)]{amari}. Here we simply
clarify that it is a special case of Theorem~\ref{thm_GHB}.
\begin{corollary}
\label{cor_HB}
If $p < \infty$ and $K^{-1} = \avg{S,S}^{-1}$ exists,
the GHB is equal to the original Helstrom bound given by
Eq.~(\ref{HB}).
\end{corollary}
\begin{proof}
  Delegated to Appendix~\ref{sec_HB}.
\end{proof}

Note that, unlike Eq.~(\ref{HB}), which assumes that $S$ consists of
linearly independent operators and $K$ is invertible,
Theorem~\ref{thm_GHB} works with no regard for any linear dependence
in $S$. This generalization is in fact indispensable to the
semiparametric theory, especially when the concept of parametric
submodels is introduced in Sec.~\ref{sec_submodels}.

The second corollary, which gives a scaling of the bound with the
number of object copies and is easy to prove via $K^{-1}$, requires
more effort to prove if $K^{-1}$ is to be avoided.
\begin{corollary}
\label{iid}
For a family of density operators that model
$N$ independent and identical quantum objects in the form of
\begin{align}
\mathbf F^{(N)}
&\equiv \BK{\rho(\theta)^{\otimes N}: \theta \in \Theta
\subseteq \mathbb R^{p}},
\end{align}
where the tensor power is defined as the tensor product
\begin{align}
\rho^{\otimes N} &\equiv \underbrace{\rho\otimes
\dots \otimes \rho}_{N \textrm{ terms}},
\end{align}
the efficient influence and the GHB are given by
\begin{align}
\delta_{\rm eff}^{(N)} &= \frac{U\delta_{\rm eff}^{(1)}}{\sqrt{N}},
&
\tilde{\mathsf{H}}^{(N)} &= \frac{\tilde{\mathsf{H}}^{(1)}}{N},
\label{scaling}
\end{align}
where $U$ is a map defined as
\begin{align}
U h &\equiv \frac{1}{\sqrt{N}}
\sum_{n=1}^N I^{\otimes(n-1)} \otimes h \otimes I^{\otimes(N-n)}.
\label{unitary}
\end{align}
\end{corollary}
\begin{proof}
Delegated to Appendix~\ref{sec_iid}.
\end{proof}

\subsection{\label{sec_gradient}Influence
operator via a functional gradient}

Theorem~\ref{thm_GHB} is useful if an influence operator
$\delta \in \mathcal D$ can be found and $\Pi(\delta|\mathcal T)$ is
tractable. One way of deriving an influence operator is to assume that
the parameter of interest is a functional $\beta[\rho]$ and consider a
derivative of $\beta[\rho]$ in the ``direction'' of an operator $h$
given by
\begin{align}
D_h\beta[\rho]
&\equiv \lim_{\epsilon \to 0}
\frac{\beta[\rho + \epsilon \rho \circ h]-\beta[\rho]}{\epsilon}.
\label{fdiff}
\end{align}
Assume that the directional derivative can be expressed as
\begin{align}
D_h\beta[\rho] = \trace (\rho \circ h) \tilde\beta
= \avg{h,\tilde\beta}
\quad
\forall h \in \mathcal Y
\label{b}
\end{align}
in terms of a $\tilde\beta \in \mathcal Y$, hereafter called a
gradient of $\beta[\rho]$. Any ordinary partial derivative of $\beta$
becomes
\begin{align}
\partial_j\beta[\rho] &= 
\lim_{\epsilon \to 0} 
\frac{\beta[\rho + \epsilon \partial_j\rho]-\beta[\rho]}{\epsilon}
= D_{S_j}\beta[\rho] = \avg{S_j,\tilde\beta}.
\end{align}
Projecting the gradient into $\mathcal Z$ then gives an influence
operator, viz.,
\begin{align}
\delta &= \Pi(\tilde\beta|\mathcal Z) = 
\tilde\beta - \Pi(\tilde\beta|\mathcal Z^\perp) =
\tilde\beta - \avg{\tilde\beta,I} \in \mathcal D,
\label{proj_b}
\end{align}
as it is straightforward to check that $\avg{\delta,I} = 0$ and
$\avg{S,\delta}= \partial\beta$.  The top flowchart in
Fig.~\ref{gradients} illustrates the steps to obtain $\delta$ from
$\beta[\rho]$.  $\tilde\beta$, $\delta$, and $\delta_{\rm eff}$ are
all gradients that satisfy Eq.~(\ref{b}); the difference lies in the
set of directions to which each is restricted. $\delta$, for instance,
is restricted to $\mathcal Z$ and orthogonal to $\mathcal Z^\perp$,
while $\delta_{\rm eff}$ is restricted to $\mathcal T$ and orthogonal
to $\mathcal T^\perp$ \footnote{More precisely, $\tilde\beta$ is the
  unique Riesz-Fr\'echet representation \cite{reed_simon} of
  $D_h\beta$ as a continuous linear functional of $h \in \mathcal Y$,
  $\delta$ is that for $h \in \mathcal Z \subset \mathcal Y$, and
  $\delta_{\rm eff}$ is that for
  $h \in \mathcal T \subseteq \mathcal Z \subset \mathcal Y$
  \cite{reed_simon,bickel93}.  The existence of each relies on
  $D_h\beta$ being continuous with respect to $h$ in each domain, so
  the existence of $\tilde\beta$ implies that of $\delta$ and
  $\delta_{\rm eff}$.}.

\begin{figure}[htbp!]
\centerline{\includegraphics[width=0.45\textwidth]{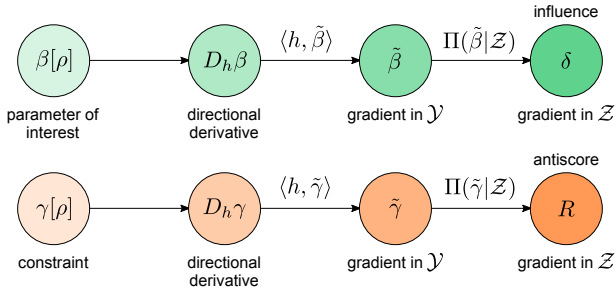}}
\caption{\label{gradients}Top (for Sec.~\ref{sec_gradient}): steps to
  obtain an influence operator $\delta$ from the functional
  $\beta[\rho]$ via Eqs.~(\ref{fdiff}), (\ref{b}), and
  (\ref{proj_b}). Bottom (for Sec.~\ref{sec_antiscore}): steps to
  obtain the antiscore operators that span $\mathcal T^\perp$ via
  Eqs.~(\ref{fdiffc}) and (\ref{proj_g}).  }
\end{figure}

Now consider some examples. The first is $\beta=\trace \rho Y$ for a
given (i.e., $\theta$-independent) observable $Y$, which leads to
\begin{align}
D_h\beta &= \trace (\rho \circ h) Y = \Avg{h,Y},
&
\delta &= Y - \beta.
\label{delta_Y}
\end{align}
The second example is the purity $\beta = \trace \rho^2$, which leads
to
\begin{align}
D_h\beta &= \trace[(\rho\circ h)\rho + \rho(\rho\circ h)] = \Avg{h,2\rho},
&
\delta &= 2(\rho - \beta).
\label{delta_purity}
\end{align}
The final example is the relative entropy
$\beta =\trace \rho (\ln\rho-\ln\sigma)$ \cite{hayashi,holevo12}.
where $\ln\rho = \sum_{j}(\ln \lambda_j)\ket{e_j}\bra{e_j}$ and
$\sigma$ is a given density operator with
$\supp(\sigma) \supseteq \supp(\rho)$. The differentiability of
$\beta$ is not a trivial question when $d = \infty$ \cite{holevo12},
but for $d < \infty$ it can be done to give
\begin{align}
D_h \beta &= \Avg{h,\ln\rho-\ln\sigma},
&
\delta &= \ln\rho-\ln\sigma - \beta,
\label{delta_entropy}
\end{align}
where $D_h\beta$ uses the fact that
$\trace \rho[\ln(\rho+\epsilon \rho \circ h)-\ln\rho]$ is second order
in $\epsilon$ for any $h \in \mathcal Z$
\cite[Theorem~6.3]{hayashi}. The von Neumann entropy is a simple
variation of this example.

\subsection{\label{sec_tangent}Projection into the tangent space}
The next step is $\Pi(\delta|\mathcal T)$. If the family of density
operators is large enough, $\mathcal T$ can fill the entire
$\mathcal Z$ and the projection becomes trivial. We call a family
full-dimensional if its tangent space at each $\rho$ satisfies
\begin{align}
\mathcal T = \mathcal Z.
\end{align}
For a specific example, consider the orthonormal basis of $\mathcal H$
given by Eq.~(\ref{Q}) and the most general parametrization of $\rho$
for $d < \infty$ given by \cite{kahn09} 
\begin{align}
\mathbf F_0 \equiv \BK{\rho(\theta) = \sum_{j} \theta_{aj} a_j + 
\sum_{k_1 < k_2}\bk{\theta_{bk} b_k + \theta_{ck} c_k}},
\label{F0}
\end{align}
where
\begin{align}
a_j &= \ket{j}\bra{j},
\\
b_k &= \frac{1}{2}
\bk{\ket{k_1}\bra{k_2} + \ket{k_2}\bra{k_1}}, \quad k_1 < k_2,
\\
c_k &= \frac{i}{2}
\bk{\ket{k_1}\bra{k_2} - \ket{k_2}\bra{k_1}}, \quad k_1 < k_2,
\end{align}
and a special entry $\theta_{a0}$ is removed from the parameters and
set as $\theta_{a0}=1- \sum_{j \neq 0}\theta_{aj}$, such that
$\trace \rho(\theta) = \sum_j\theta_{aj} = 1$ and 
\begin{align}
p = d^2 -1.
\end{align}
$\partial\rho$ is then given by
\begin{align}
\partial_{aj} \rho &= a_j - a_0,
&
\partial_{bk} \rho &= b_k,
&
\partial_{ck} \rho &= c_k.
\label{drho}
\end{align}
The next theorem is a key step in deriving simple analytic results.
\begin{theorem} The $\mathbf F_0$ family is full-dimensional.
\label{TeqZ}
\end{theorem}
\begin{proof}
  Consider the solution to $\avg{S,h} = 0$ for an $h\in \mathcal Z$.
  All operators are bounded if $d < \infty$.  We can then use
  Eqs.~(\ref{cyclic}) and (\ref{drho}) to obtain
\begin{align}
\trace(\partial_{aj}\rho) \hat h &= 
\bra{j}\hat h\ket{j} - \bra{0}\hat h\ket{0} = 0,
\\
\trace(\partial_{bk}\rho) \hat h &= 
\real \bra{k_1}\hat h\ket{k_2} = 0,
\quad
k_1 < k_2,
\\
\trace(\partial_{ck}\rho) \hat h &= \imag \bra{k_1}\hat h\ket{k_2} = 0,
\quad
k_1 < k_2,
\end{align}
where $\hat h$ is any operator in the equivalence class of $h$. Thus
all the diagonal entries of $\hat h$ are equal to
$\bra{0}\hat h\ket{0}$, and all the off-diagonal entries are zero. In
other words, $\hat h = \bra{0}h\ket{0} \hat I$, where $\hat I$ is the
identity operator.  But $h \in \mathcal Z$ also means that
$\trace \rho \hat h = \bra{0}\hat h\ket{0} = 0$, resulting in
$\hat h = 0$ as the only solution. Hence $\mathcal T^\perp = \{0\}$
contains only the zero element, and $\mathcal T = \mathcal Z$.

\end{proof}

$\mathbf F_0$ implies that the experimenter knows nothing about the
density operator, apart from the Hilbert space $\mathcal H$ on which
it operates. Despite the high dimension of the family,
Theorems~\ref{thm_GHB} and \ref{TeqZ} turn the problem into a
trivial exercise once an influence operator has been found, since a
$\delta \in \mathcal D \subseteq \mathcal Z$ is already in
$\mathcal Z = \mathcal T$ and hence efficient.  Corollary~\ref{iid}
can then be used to extend the result for $N$ copies. For
$\beta = \trace \rho Y$, Eq.~(\ref{delta_Y}) leads to
\begin{align}
\tilde{\mathsf{H}}^{(N)} &= \frac{\norm{\delta}^2}{N}
= \frac{1}{N}\trace \rho\bk{Y - \beta}^2.
\label{GHB_Y}
\end{align}
This implies that a von Neumann measurement of $Y$ of each copy and
taking the sample mean of the outcomes are already efficient; no other
measurement can do better in terms of unbiased estimation. For
$\beta = \trace \rho^2$, Eq.~(\ref{delta_purity}) leads to
\begin{align}
\tilde{\mathsf{H}}^{(N)} &= 
\frac{\norm{\delta}^2}{N}= \frac{4}{N}\trace \rho\bk{\rho-\beta}^2,
\label{GHB_purity}
\end{align}
and for $\beta = \trace \rho (\ln\rho- \ln\sigma)$,
Eq.~(\ref{delta_entropy}) leads to
\begin{align}
\tilde{\mathsf{H}}^{(N)} &= 
\frac{\norm{\delta}^2}{N}=\frac{1}{N}
\trace \rho \bk{\ln\rho- \ln\sigma - \beta}^2.
\label{GHB_entropy}
\end{align}
Intriguingly, this expression coincides with the information variance
that has found uses in other contexts of quantum information theory,
such as quantum hypothesis testing
\cite{tomamichel13,*li14a,*tomamichel16}.

Deriving Eqs.~(\ref{GHB_Y})--(\ref{GHB_entropy}) via the conventional
brute-force method would entail the following steps:
\begin{enumerate}
\item Assume the $\mathbf F_0$ family of density operators given by
  Eq.~(\ref{F0}), with $p = d^2 - 1$ parameters.

\item Compute the $p$ score operators via Eq.~(\ref{sld}).

\item Compute the $p$-by-$p$ Helstrom information matrix $K$ via
  Eq.~(\ref{K}).

\item Compute the inverse $K^{-1}$.

\item Compute $\beta(\theta)$ via Eq.~(\ref{F0}),
$\partial\beta(\theta)$, and the Helstrom bound via Eq.~(\ref{HB}).
\end{enumerate}
While this method has been used to produce Eq.~(\ref{GHB_Y})
\cite{watanabe10}, it is less clear whether it can easily give
Eq.~(\ref{GHB_purity}) or Eq.~(\ref{GHB_entropy}).  
  Contrast the brute-force method with the proposal here:
\begin{enumerate}
\item Compute the influence operator $\delta$ via a
functional derivative of $\beta[\rho]$ according to Sec.~\ref{sec_gradient}.

\item Find the tangent space $\mathcal T$ of the density-operator
  family or the orthocomplement $\mathcal T^\perp$.  For example,
  Theorem~\ref{TeqZ} shows that $\mathcal T$ is full-dimensional for
  the family of arbitrary density operators, while Sec.~\ref{sec_con}
  later shows that $\mathcal T^\perp$ may remain tractable for smaller
  families.

\item Compute
  $\delta_{\rm eff} = \Pi(\delta|\mathcal T) = \delta-
  \Pi(\delta|\mathcal T^\perp)$ and
  $\tilde{\mathsf{H}} = \norm{\delta_{\rm eff}}^2 = \trace \rho
  \delta_{\rm eff}^2$.
\end{enumerate}
Each step is tractable for all the examples here, regardless of the
dimensions.

Equations~(\ref{GHB_Y})--(\ref{GHB_entropy}) are the quantum bounds
promised in Sec.~\ref{sec_preview}, although they are merely the
simplest examples of what the semiparametric methodology can offer, as
Secs.~\ref{sec_con}--\ref{sec_displacement} later show.

\section{\label{sec_submodels}Parametric submodels}

The proof of Theorem~\ref{TeqZ} works only in the finite-dimensional
case ($p = d^2-1 < \infty$).  For infinite-dimensional problems, the
beautiful concept of parametric submodels
\cite{stein56,ibragimov81,bickel93,tsiatis06} offers a more rigorous
approach. Let 
\begin{align}
\mathbf G \equiv \BK{\rho(g): g \in \mathcal G}
\label{mother}
\end{align}
be a ``mother'' density-operator family, where $\mathcal G$ may be an
infinite-dimensional space. The density operators are still assumed to
operate on a common separable Hilbert space $\mathcal H$. Denote the
true density operator in the family as $\rho$.  A parametric submodel
$\mathbf F^\sigma$ is defined as any subset of $\mathbf G$ that
contains the true $\rho$ and has the parametric form of
Eq.~(\ref{family}). To wit,
\begin{align}
\mathbf F^\sigma &\equiv \BK{\sigma(\theta): \theta \in \Theta^\sigma
\subseteq \mathbb R^s, 
\sigma(\phi) = \rho} \subseteq \mathbf G,
\label{submodel}
\end{align}
where $s$ denotes the dimension of the parameter and $\phi$ denotes
the parameter value at which $\sigma(\phi) = \rho$ is the truth; both
may be specific to the submodel.  In the language of geometry
\cite{hayashi,hayashi05,uhlmann93}, each $\mathbf F^\sigma$ is an
$s$-dimensional surface in $\mathbf G$, and all the surfaces are
required to intersect at $\rho$.  Figure~\ref{manifolds} illustrates
the concept.

\begin{figure}[htbp!]
\centerline{\includegraphics[width=0.35\textwidth]{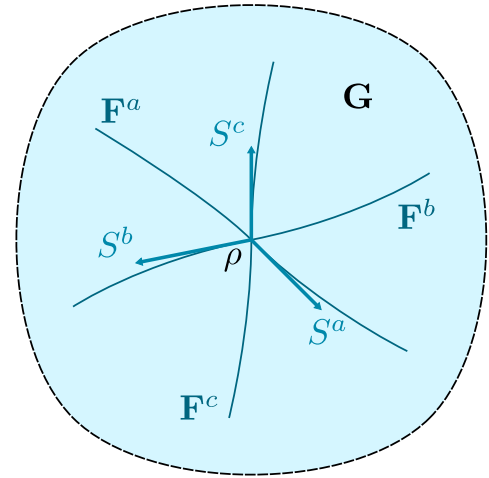}}
\caption{\label{manifolds} The space represents $\mathbf G$, a mother
  family of density operators. The true density operator is denoted as
  $\rho$. Parametric submodels are represented by curves in
  $\mathbf G$ that intersect at $\rho$. Each score $S^\sigma$ is a
  tangent vector that quantifies the ``velocity'' of a
  density-operator trajectory in a certain direction.}
\end{figure}

Each submodel $\mathbf F^\sigma$ is assumed
to be smooth enough for scores to be defined in the
same way as before by
\begin{align}
(\partial\sigma)_{\theta = \phi} &= \rho \circ
S^\sigma,
\end{align}
which denotes a system of $s$ equations given by
\begin{align}
\left.\partial_j\sigma(\theta)\right|_{\theta = \phi}
&= \sigma(\phi) \circ S_j^\sigma(\phi) = \rho \circ S_j^\sigma(\phi).
\end{align}
 As everything is evaluated at the true $\rho$, the scores of all
  submodels in fact live in the same Hilbert space $\mathcal Z$ with
  respect to $\rho$.  Let the set of \emph{all} parametric submodels
  of $\mathbf G$ with respect to the truth be
\begin{align}
\mathcal F &\equiv \BK{\mathbf F^\sigma: \sigma \in \mathcal S},
\end{align}
where $\mathcal S$ denotes the set of indices that label all
  the submodels.  Define the tangent set as the set of the scores
from all such parametric submodels of $\mathbf G$, viz.,
\begin{align}
\BK{S} &\equiv \bigcup_{\sigma \in \mathcal S} \BK{S^\sigma},
\end{align}
and the tangent space as the span of the set, viz.,
\begin{align}
\mathcal T &\equiv \cspn\BK{S} \subseteq \mathcal Z.
\label{tangent_mother}
\end{align}
An influence operator is now defined as any operator that satisfies
the unbiasedness condition for all submodels with respect to
$\{S\}$.  The condition can be expressed as
\begin{align}
\Avg{S^\sigma,\delta} = (\partial\beta)_{\theta = \phi}\quad 
\forall \mathbf F^\sigma \in \mathcal F,
\label{local_unbiased}
\end{align}
where $(\partial\beta)_{\theta=\phi}$ is specific to each submodel.
If $\avg{S,\delta} = \partial\beta$ in Eq.~(\ref{influence_set}) is
taken to mean Eq.~(\ref{local_unbiased}), then the influence-operator
set $\mathcal D$ is still defined by Eq.~(\ref{influence_set}).  The
error operator given by Eq.~(\ref{error}) for an unbiased measurement
still satisfies Eq.~(\ref{local_unbiased}) by the generic arguments in
Ref.~\cite[Sec.~6.2]{holevo11}, which apply to any submodel, so the
error operator remains in $\mathcal D$, and Eq.~(\ref{GHB}) still
holds. Theorem~\ref{thm_GHB} can now be extended for the mother
family.

\begin{theorem}
\label{thm_SHB}
The GHB in Eq.~(\ref{GHB}) for the mother family $\mathbf G$ is given
by
\begin{align}
\tilde{\mathsf{H}} = \min_{\delta \in \mathcal D} \norm{\delta}^2
= \norm{\delta_{\rm eff}}^2,
\end{align}
where the efficient influence $\delta_{\rm eff}$ is the unique element
in the influence-operator set $\mathcal D$ given by
\begin{align}
\delta_{\rm eff} &= \Pi(\delta|\mathcal T),
\end{align}
$\delta$ is any influence operator in $\mathcal D$, and $\mathcal T$
is the tangent space spanned by the scores of all parametric submodels
of $\mathbf G$.
\end{theorem}
\begin{proof}
  The proof is identical to that of Theorem~\ref{thm_GHB} if one
  takes $\{S\}$ to be the tangent set containing the scores of all
  parametric submodels.
\end{proof}

Corollary~\ref{iid} can also be generalized in an almost identical
way, although the proof requires more careful thought.
\begin{corollary}
\label{iid2}
For a family of density operators that model $N$ independent and
identical quantum objects in the form of
\begin{align}
\mathbf G^{(N)} \equiv \BK{\rho(g)^{\otimes N}: g\in\mathcal G},
\label{GN}
\end{align}
the efficient influence and the GHB are given by
\begin{align}
\delta_{\rm eff}^{(N)} &= \frac{ U\delta_{\rm eff}^{(1)}}{\sqrt{N}},
&
\tilde{\mathsf{H}}^{(N)} &= \frac{\tilde{\mathsf{H}}^{(1)}}{N},
\label{scaling2}
\end{align}
where $\delta_{\rm eff}^{(1)}$ and $\tilde{\mathsf{H}}^{(1)}$ are those for the
$N = 1$ family according to Theorem~\ref{thm_SHB} and $U$ is the
map given by Eq.~(\ref{unitary}).
\end{corollary}
\begin{proof}
Delegated to Appendix~\ref{sec_iid2}.
\end{proof}

We now generalize Theorem~\ref{TeqZ} for infinite-dimensional
systems. This is also a more precise generalization of a classic
result in semiparametric theory \cite[Example~1 in
Sec.~3.2]{bickel93}.

\begin{theorem}
  $\mathbf G_0$, defined as the family of arbitrary density operators,
  is full-dimensional.
\label{TeqZ2}
\end{theorem}

\begin{proof}
  We call a Hilbert-space element in $\mathcal Y$ bounded and denote it
  by $\opnorm h < \infty$ if its equivalence class contains a bounded
  operator $\hat h$. Denote the set of all bounded elements in
  $\mathcal Z$ as
\begin{align}
\mathcal B \equiv \BK{h \in \mathcal Z: \opnorm h < \infty}.
\label{set_bounded}
\end{align}
Take any $h \in \mathcal B$ and its bounded operator $\hat h$.
Construct a scalar-parameter exponential family as
\cite{hayashi,hayashi05}
\begin{align}
\sigma(\theta) &= \frac{\kappa(\theta)}{\trace\kappa(\theta)},
&
\kappa(\theta) &= \exp(\theta \hat h/2)
\rho \exp(\theta \hat h/2),
\label{exp_family}
\end{align}
where $\theta \in \mathbb R$ and the truth is at $\sigma(0) = \rho$.
As
$\hat h$ is bounded, $\exp(\theta \hat h /2)$ is bounded and strictly
positive.  As $\rho$ is nonnegative and unit-trace, $\kappa(\theta)$
is nonnegative and trace-class
\cite[Theorem~2.7.2]{holevo11}. Moreover, $\trace\kappa(\theta)$
satisfies the properties
\begin{align}
\infty
> \trace\kappa(\theta) = \trace \rho \exp(\theta \hat h) > 0,
\end{align}
because $\kappa(\theta)$ is trace-class and $\exp(\theta\hat h)$ is
strictly positive. Hence $\sigma(\theta)$ is a valid density operator at
any $\theta$.  Since $\mathbf G_0$ contains arbitrary density
operators,
$\mathbf F^\sigma = \{\sigma(\theta): \theta \in \mathbb R, \sigma(0) =
\rho\}$ is a parametric submodel of $\mathbf G_0$.  It is
straightforward to show that
\begin{align}
(\partial\sigma)_{\theta = 0} &= \sigma(0) \circ \hat h = \rho \circ \hat h,
\end{align}
so the score for this model can be taken as $S^\sigma = h$.

Define a submodel in the same way for \emph{every} $h \in \mathcal B$,
such that all of the $\mathcal B$ elements are in the tangent set
$\{S\}$, leading to $\mathcal B \subseteq \{S\} \subseteq \mathcal
T$. As $\mathcal T$ is closed, the limit points of $\mathcal B$ must
also be in $\mathcal T$, and
$\overline{\mathcal B} \subseteq \mathcal T$, where
$\overline{\mathcal B}$ is the closure of $\mathcal B$.
Lemma~\ref{cauchy} in Appendix~\ref{sec_cauchy} states that
$\mathcal B$ is a dense subset of $\mathcal Z$, so
\begin{align}
\mathcal Z = \overline{\mathcal B} \subseteq \mathcal T.
\end{align}
Together with the fact $\mathcal T \subseteq \mathcal Z$, this implies
$\mathcal T = \mathcal Z$, and the theorem is proved.

\end{proof}

A comparison of the proofs of Theorems~\ref{TeqZ} and \ref{TeqZ2}
shows how the parametric-submodel concept works. Instead of dealing
with one large family such as Eq.~(\ref{F0}), here one exploits the
freedom offered by $\mathbf G_0$ to specify many ad-hoc and elementary
submodels. Each submodel in the proof cannot be simpler---the
exponential family is simply a type of geodesics through $\rho$ in
density-operator space \cite{hayashi}.  In fact, we do not have to use
the exponential family, and other families may also be used as long as
they fit the purpose of the proof.  An enormous number of submodels
are introduced, one for each $\mathcal B$ element in the proof,
leading to an extremely overcomplete tangent set. But that presents no
trouble for the geometric approach; only the resultant tangent space
matters at the end.  Figure~\ref{straight_lines} illustrates the idea.

\begin{figure}[htbp!]
  \centerline{\includegraphics[width=0.35\textwidth]{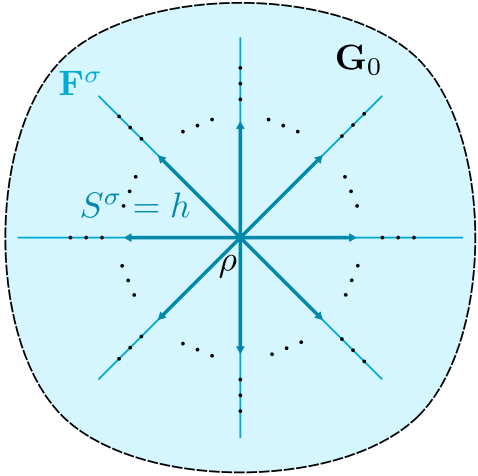}}
  \caption{\label{straight_lines} For any $h \in \mathcal B$, one can
    associate with it an exponential family (a straight line in
    the density-operator space) that passes through $\rho$. Since
    $\mathbf G_0$ contains arbitrary density operators, every line
    must be contained in $\mathbf G_0$. It follows that each line is a
    parametric submodel for $\mathbf G_0$, and each $h$ should be put
    in the tangent set. The dots represent the fact that the proof
    involves lines in all directions and, on each line, scores with
    all possible norms.}
\end{figure}

By virtue of Theorem~\ref{TeqZ2}, an influence operator
$\delta \in \mathcal D \subseteq \mathcal Z = \mathcal T$ found for a
parameter of interest is the efficient one for $\mathbf G_0$. The
examples in Secs.~\ref{sec_gradient} and \ref{sec_tangent} work for
$\mathbf G_0$ in the same way they work for $\mathbf F_0$.  If $\beta$
is given by $\beta[\rho]$, an influence operator that satisfies
Eq.~(\ref{local_unbiased}) can be found via a gradient of
$\beta[\rho]$, as shown in Sec.~\ref{sec_gradient} and
Fig.~\ref{gradients}. In particular, the influence operators given by
Eqs.~(\ref{delta_Y})--(\ref{delta_entropy}) and the bounds given by
Eqs.~(\ref{GHB_Y})--(\ref{GHB_entropy}) for the various examples
should still hold for $\mathbf G_0$, although the entropy example may
require a more rigorous treatment when $d = \infty$ \cite{holevo12}.

\section{\label{sec_con}Constrained bounds}

\subsection{\label{sec_antiscore}Antiscore operators}
Consider a constrained family of density operators defined as
\begin{align}
\mathbf G_\gamma &\equiv \BK{\rho(g) \in \mathbf G_0: \gamma[\rho(g)] = 0},
\label{family_con}
\end{align}
where $\gamma[\rho(g)] = 0$ denotes a finite set of equality
constraints $\{\gamma_k[\rho(g)] = 0: k = 1,\dots,r\}$. Such
constraints appear often in quantum thermodynamics
\cite{jaynes57,gogolin16}.  If there exist gradient operators
$\{\tilde\gamma_k \in \mathcal Y\}$ such that, for any
$h \in \mathcal Y$,
\begin{align}
D_h \gamma_k[\rho] &= \Avg{h,\tilde\gamma_k},
\label{fdiffc}
\end{align}
then each operator given by
\begin{align}
R_k \equiv \Pi(\tilde\gamma_k|\mathcal Z) = \tilde\gamma_k 
- \Avg{\tilde\gamma_k,I} \in \mathcal Z
\label{proj_g}
\end{align}
satisfies
\begin{align}
D_h\gamma_k[\rho]
 &= \Avg{h,R_k}\quad \forall h \in \mathcal Z,
\label{Rgradient}
\end{align}
and the constraint $\gamma[\rho(g)] = 0$ implies that
$\partial\gamma_k[\rho] = \Avg{S^\sigma,R_k} = 0$ for all submodels
and $k$.  In short, we write
\begin{align}
  \partial\gamma[\rho] = \Avg{S,R} = 0.
\label{R}
\end{align}
Thus $\{R\}$ is orthogonal to the tangent set $\{S\}$ and $\spn\{R\}$
must be a subset of $\mathcal T^\perp$. We call $R$ the antiscore
operators, as the following theorem shows that they span
$\mathcal T^\perp$ in the same way the scores span $\mathcal T$.

\begin{theorem}
  If $\avg{R,R}^{-1}$ exists, $\mathcal T^\perp = \spn\{R\}$ for the
  $\mathbf G_\gamma$ family.
\label{thm_R}
\end{theorem}
\begin{proof}
The proof again follows the classical case \cite[Example~3 in
  Sec.~3.2]{bickel93}.  Let
\begin{align}
\mathcal R &\equiv \spn\BK{R},
&
\mathcal R^\perp &\equiv \BK{h \in \mathcal Z: \Avg{R,h} = 0}.
\end{align}
In view of Eq.~(\ref{R}),
\begin{align}
\mathcal T &\subseteq \mathcal R^\perp.
\label{TRperp}
\end{align}
Now construct a parametric submodel $\mathbf F^\sigma$ in terms of
each $h \in \mathcal R^\perp$ as
\begin{align}
\sigma(\theta) &= 
\frac{\kappa(\theta)}{\trace \kappa(\theta)},
&
\kappa(\theta) &= f(\theta h+\theta g)\rho f(\theta h+\theta g),
\label{f_family}
\end{align}
where $\theta \in \mathbb R$, $g = w^\top R\in \mathcal R$ is an
operator to be specified later, and $f(u)$ is defined with respect to
the spectral representation of $u = \int \lambda dE_u(\lambda)$ as
\begin{align}
f(u) &= \int \Bk{1 + \tanh\bk{\frac{\lambda}{2}}} dE_u(\lambda).
\label{f}
\end{align}
$f(u)$ is bounded and positive even if $u$ is unbounded, so
$\sigma(\theta)$ is a valid density operator. Since
$\rho \in \mathbf G_\gamma$, $\gamma[\rho] = 0$.  For a $\sigma(\theta)$
away from $\rho$ with $\theta \neq 0$,
\begin{align}
\gamma[\sigma(\theta)] &= \gamma[\rho] + 
\theta D_{h+g} \gamma[\rho] +  o(\theta)
\\
&= \theta \Avg{R,h+g} + o(\theta)
\label{gradientstep}
\\
&= \theta\Avg{R,g} + o(\theta),
\label{lin_gamma}
\end{align}
where Eq.~(\ref{gradientstep}) uses Eq.~(\ref{Rgradient}) and the last
step uses the fact $h \in \mathcal R^\perp$. To make
$\sigma(\theta)$ satisfy the constraint $\gamma[\sigma(\theta)] = 0$,
$g(\theta) = w(\theta)^\top R $ can be set as a function of $\theta$ to
cancel the $o(\theta)$ term, with
\begin{align}
 w(\theta) &= -\Avg{R,R}^{-1} o(\theta)/\theta.
\label{root}
\end{align}
Then $\gamma[\sigma(\theta)] = 0$ and $\mathbf F^\sigma$ is a valid
parametric submodel of $\mathbf G_\gamma$. Equation~(\ref{root}) also
implies that $\theta g(\theta) = o(\theta)$ is negligible relative to
$\theta h$ for infinitesimal $\theta$, so the score for
$\mathbf F^\sigma$ is $h$, which should be put in the tangent set
$\{S\}$. As this procedure can be done for any
$h \in \mathcal R^\perp$,
$\mathcal R^\perp \subseteq \{S\} \subseteq \mathcal T$.  Together
with Eq.~(\ref{TRperp}), this leads to
$\mathcal T = \mathcal R^\perp$, giving
$\mathcal T^\perp = \mathcal R$.

\end{proof}

The family given by Eqs.~(\ref{f_family}) and (\ref{f}) is more
  convenient to use here than the exponential family used in the proof
  of Theorem~\ref{TeqZ2}.  The $f(u)$ defined by Eq.~(\ref{f}) is a
  generalization of the classical version in Ref.~\cite[Example~1 in
  Sec.~3.2]{bickel93} and plotted in Fig.~\ref{smooth_f}. It is
  designed to give a valid density operator via
  Eqs.~(\ref{f_family})---even if the argument is an unbounded
  operator---yet produce the desired score when linearized at
  $\theta = 0$. An adjustable operator $g(\theta)$ is
included in the submodel to make $\sigma(\theta)$ satisfy the
constraint away from $\rho$.  Figure~\ref{constrained_tangent} further
illustrates the idea of the proof.

\begin{figure}[htbp!]
\centerline{\includegraphics[width=0.4\textwidth]{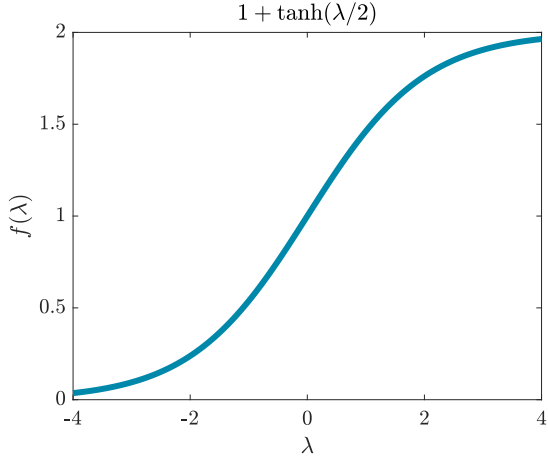}}
\caption{\label{smooth_f}A plot of $f(\lambda) = 1 + \tanh(\lambda/2)$ to
  illustrate its boundedness and positivity.}
\end{figure}

\begin{figure}[htbp!]
\centerline{\includegraphics[width=0.4\textwidth]{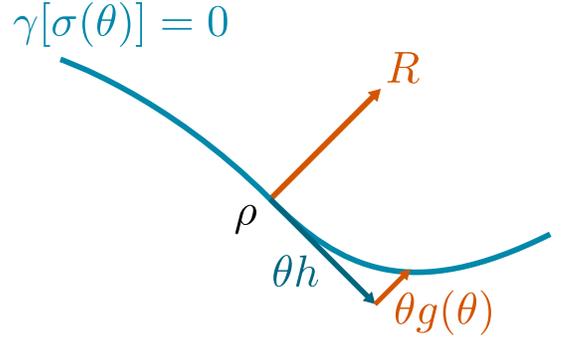}}
\caption{\label{constrained_tangent} Each $R$ is a vector normal to
  the surface defined by $\gamma[\rho(g)] = 0$ in density-operator
  space.  For any $h \in \mathcal R^\perp$, a parametric submodel
  $\sigma(\theta)$ can be constructed to satisfy the constraint
  $\gamma[\sigma(\theta)] = 0$. Away from $\rho$, a correction
  $\theta g(\theta) = o(\theta)$ in $\mathcal R$ is needed to make
  $\sigma(\theta)$ stay with the constraint. The tangent vector of the
  submodel at $\rho$ is still $h$, since the correction is higher
  order in $\theta$.}
\end{figure}

Given an influence operator $\delta$, such as those derived in
Sec.~\ref{sec_gradient}, the efficient influence and the GHB can be
computed in terms of $\mathcal T^\perp$ instead of $\mathcal T$ via
\begin{align}
\delta_{\rm eff} &= \Pi(\delta|\mathcal T) 
= \delta - \Pi(\delta|\mathcal T^\perp),
\label{eff_subtract}
\\
\norm{\delta_{\rm eff}}^2 
&= \norm{\delta}^2 - \norm{\Pi(\delta|\mathcal T^\perp)}^2.
\label{GHB_subtract}
\end{align}
The same projection formula that gives $\delta_{\rm eff}$ in
Appendix~\ref{sec_HB} can be adapted to give
\begin{align}
\Pi(\delta|\mathcal T^\perp) &= \Avg{R,\delta}^\top\Avg{R,R}^{-1}R,
\label{holes}
\\
\norm{\Pi(\delta|\mathcal T^\perp)}^2 &= 
\Avg{R,\delta}^\top \Avg{R,R}^{-1} \Avg{R,\delta}.
\label{reduction}
\end{align}
Equations~(\ref{holes}) and (\ref{reduction}) remain tractable if the
constraints are few.  The gradients of $\gamma[\rho]$ can be derived
in the same way as those of $\beta[\rho]$, as shown in
Fig.~\ref{gradients}, and $R$ can be computed analytically for linear
constraints, the purity constraint, and the entropy constraint by
following the same type of calculations shown in
Eqs.~(\ref{delta_Y})--(\ref{delta_entropy}). Equation~(\ref{correlation_bound})
is a special example of the constrained GHB when
$\beta = \trace\rho Y$ and $\gamma = \trace \rho(Z-\zeta) = 0$.

\subsection{Entropy estimation in quantum thermodynamics}
In quantum thermodynamics, conserved quantities of a dynamical system,
such as the energy and the particle number, are expressed as moment
constraints on the density operator with respect to a vector of
observables $Z$ and a vector of constants $\zeta$, viz.,
\begin{align}
\trace \rho Z_k &= \zeta_k,
\quad
k =1,\dots,r.
\end{align}
Given such constraints, the density
operator is often assumed to be the one with the maximum entropy
\cite{jaynes57}, known as the generalized Gibbs ensemble
\cite{gogolin16}. Such an assumption, however, requires verification
and does not hold out of equilibrium.  Experiments on Bose gases have
been performed to study the quantum states at different times and the
validity of the maximum-entropy principle at steady state
\cite{kinoshita06,langen15,langen15a}.

When the maximum-entropy principle is in question for those
experiments, it is prudent to make no prior assumption about the
density operator other than the constraints.  Thus one should consider
a family of density operators given by Eq.~(\ref{family_con}), where
the vectoral constraint is
$\gamma[\rho] = \trace \rho(Z - \zeta) = 0$.  Suppose that the von
Neumann entropy $\beta = -\trace \rho \ln \rho$ is the parameter of
interest.  The estimation of $\beta$ is then a problem of quantum
semiparametric estimation.

As the experiments typically involve high-dimensional systems, quantum
state tomography is impractical. More efficient estimation of $\beta$
should exist. The formalism here leads to a quantum limit given by
\begin{align}
\mathsf{E} &\ge \frac{1}{N}\bk{\norm{\delta}^2
- \Avg{R,\delta}^\top\Avg{R,R}^{-1}\Avg{R,\delta}},
\\
\delta &= -\ln \rho -\beta,
\quad
R = Z - \zeta.
\end{align}
This bound is equivalent to the Holevo bound, as shown in
Sec.~\ref{sec_multi}, so it is asymptotically attainable in principle,
at least for finite-dimensional systems
\cite{kahn09,gill_guta,demkowicz20}, although the experimental
implementation of efficient measurements remains an open question.

As entropy is an excellent measure of randomness and a central
quantity in information theory, entropy estimation has many
applications beyond thermodynamics. In classical statistics, the
semiparametric estimation of entropic quantities is a well studied
problem with known near-efficient estimators and applications in
universal coding, statistical tests, random-number generation,
econometrics, spectroscopy, and even neuroscience
\cite{beirlant97,*paninski03,*cover}. In the quantum domain, one
application is universal quantum information compression
\cite{jozsa98}: knowing just the von Neumann entropy and nothing else
about $\rho$ allows the quantum information to be compressed in
accordance with the entropy. Another application is the estimation of
an entropic measure of entanglement, which allows one to demonstrate
entanglement without full tomography \cite{guehne09}. The quantum
limit here quantifies the minimum amount of resources needed to
achieve a desired precision. Its asymptotic attainability suggests
that it is a lofty but fair yardstick for experimental design.


\subsection{\label{sec_phil}Philosophy}
The proposed approach to quantum semiparametric bounds is the polar
opposite of the usual approach in quantum metrology. In the usual
bottom-up approach, one assumes a small family of density operators
with a few parameters and computes $\norm{\Pi(\delta|\mathcal T)}^2$
that is determined by the overlap between $\delta$ and the scores $S$.
Here, one starts with a large family with almost full dimension,
computes $\norm{\delta}^2$ for an amenable $\delta$, and then reduces
it by $\norm{\Pi(\delta|\mathcal T^\perp)}^2$ that is determined by
the overlap between $\delta$ and the antiscores $R$, as illustrated by
Fig.~\ref{philosophy}. The complexity of the problem thus depends on
the dimension of the family, and the essential insight of this work is
that the problem can become simple again when the dimension is close
to being full.  Of course, if the dimension of $\mathcal T^\perp$ is
high, the top-down approach may also suffer from the curse of
dimensionality. The medium families with both $\mathcal T$ and
$\mathcal T^\perp$ in high dimensions are the most difficult to deal
with, as they may be impregnable from either end.

\begin{figure}[htbp!]
\centerline{\includegraphics[width=0.45\textwidth]{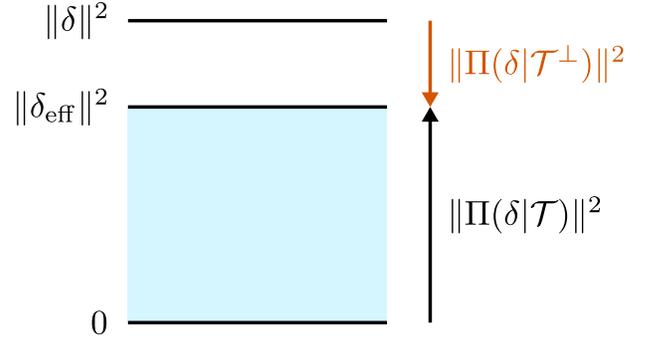}}
\caption{\label{philosophy}An illustration of the conventional
  bottom-up approach to quantum bounds and the top-down approach to
  semiparametric bounds, as discussed in Sec.~\ref{sec_phil}.}
\end{figure}

\subsection{Looser bounds}
It may often be the case that, despite one's best efforts, the exact
$\delta_{\rm eff}$ for a problem remains intractable. Then a standard
strategy in statistics and quantum metrology is to sandwich
$\norm{\delta_{\rm eff}}^2$ between upper and lower bounds.
$\norm{\delta}^2$ is an obvious upper bound and can be obtained from
the gradient method in Sec.~\ref{sec_gradient} if $\beta$ can be
expressed as a functional $\beta[\rho]$. Another way is to use
Eq.~(\ref{MSE}) if an unbiased measurement and its error are
known. The evaluation of lower bounds, on the other hand, can be
facilitated by the following proposition.

\begin{proposition}
\label{loose}
Let $\mathcal V \subseteq \mathcal T$ be a closed subspace of
$\mathcal T$ and $\mathcal V^\perp$ be the orthocomplement of
$\mathcal V$ in $\mathcal Z$. Then
\begin{align}
\tilde{\mathsf H} &= 
\norm{\delta_{\rm eff}}^2 \ge \norm{\Pi(\delta|\mathcal V)}^2
= \norm{\delta}^2 - \norm{\Pi(\delta|\mathcal V^\perp)}^2.
\end{align}
In particular, if 
\begin{align}
\mathcal V = \cspn\BK{S^\sigma}
\end{align}
is taken as the tangent space for a particular parametric submodel
$\mathbf F^\sigma$, then 
\begin{align}
\norm{\Pi(\delta|\mathcal V)}^2 &= \tilde{\mathsf{H}}^\sigma
\label{GHB_sub}
\end{align}
is the GHB for that submodel.
\end{proposition}
\begin{proof}
Delegated to Appendix~\ref{sec_loose}.
\end{proof}
A tight lower bound on $\norm{\delta_{\rm eff}}^2$ can be sought by
devising a submodel that is as unfavorable to the estimation of
$\beta$ as possible. Another approach is to devise an overconstrained
model with $\mathcal V^\perp \supseteq \mathcal T^\perp$ and evaluate
a lower bound on $\norm{\delta_{\rm eff}}^2$ from the top by
overshooting, as illustrated by Fig.~\ref{loose_bound}.

\begin{figure}[htbp!]
\centerline{\includegraphics[width=0.45\textwidth]{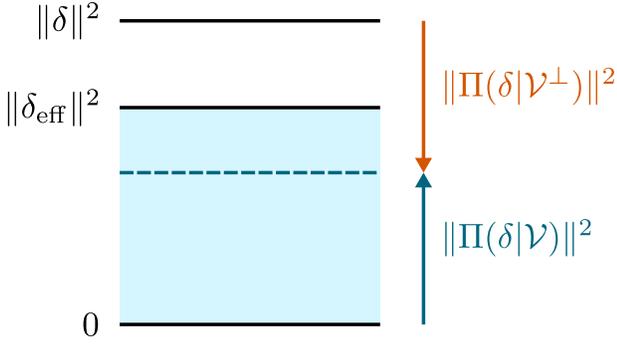}}
\caption{\label{loose_bound} One can obtain a lower bound on
  $\norm{\delta_{\rm eff}}^2$ either by undershooting from the bottom
  via a more amenable subspace $\mathcal V \subseteq \mathcal T$, or
  overshooting from the top via an overconstrained model with
  $\mathcal V^\perp \supseteq \mathcal T^\perp$.}
\end{figure}

\section{\label{sec_optics}Examples in optics}

\subsection{\label{sec_quadrature}Quadrature estimation}
Here we further illustrate the theory with examples in optics, where
quantum measurement theory has found the most experimental success
\cite{wiseman_milburn}. For the first and simplest example, let $\rho$
be a density operator of an optical mode and assume the $\mathbf G_0$
family of arbitrary density operators.  Consider the estimation of the
mean of a quadrature operator $Y$, with $\beta = \trace \rho Y$. This
problem appears often in optical state characterization,
communication, and sensing, where $\beta$ is a displacement parameter
\cite{paris04}. The GHB is given by Eq.~(\ref{GHB_Y}), and homodyne
detection of $Y$ is efficient.  Note that this example is different
from all previous studies of quadrature estimation
\cite{helstrom,holevo11}, which assume Gaussian states or similarly
low-dimensional parametric models. The semiparametric scenario here
allows $\rho$ to be arbitrary and possibly non-Gaussian.

Now suppose that side information $\trace \rho Z = \zeta$ concerning
another quadrature $Z$ is available. It follows from
Sec.~\ref{sec_con} that the efficient influence is now
\begin{align}
\delta_{\rm eff}&= Y-\beta - \frac{C_{YZ}}{V_Z}(Z-\zeta),
\end{align}
where $C_{YZ}$ and $V_Z$ are given by Eqs.~(\ref{correlation_bound2}).
The GHB is then given by Eq.~(\ref{correlation_bound}), which is
lowered by any correlation between $Y$ and $Z$.  From the efficient
influence, one may use Eq.~(\ref{error}) to find an efficient
measurement, which obeys
\begin{align}
\int \check\beta(\lambda) dE(\lambda)
&= Y - \frac{C_{YZ}}{V_Z}(Z-\zeta).
\end{align}
This can be satisfied if the POVM measures the quadrature
$Y - (C_{YZ}/V_Z) Z$ instead of the obvious $Y$.  Notice, however,
that $C_{YZ}/V_Z$ depends on the unknown $\rho$. Whether adaptive
measurements \cite{wiseman_milburn} can implement this POVM
approximately and whether asymptotic attainability is possible for
this infinite-dimensional problem are interesting open questions. One
approach may be to form rough estimates of the covariances $C_{YZ}$
and $V_Z$ via heterodyne detection of a portion of the light first,
and then measure the desired quadrature via homodyne detection based
on the approximate $C_{YZ}/V_Z$.

\subsection{\label{sec_classical}Family of classical states}
For a more nontrivial example, consider a density-operator family in
the form
\begin{align}
\mathbf G_c
&\equiv \BK{\rho(P) = \int d^2\alpha P(\alpha)\ket{\alpha}\bra{\alpha}:
P \in \mathcal G},
\\
\mathcal G &= \textrm{all positive probability densities},
\label{Gclassical}
\end{align}
where $\alpha = \alpha'+i\alpha'' \in \mathbb C$,
$d^2\alpha = d\alpha'd\alpha''$, $\ket{\alpha}$ is a coherent state,
and $P$ is the Glauber-Sudarshan function \cite{mandel}.  As $P$ is
assumed to be positive, $\mathbf G_c$ is a family of classical states
\cite{mandel} and a strict subset of $\mathbf G_0$. The assumption of
$\mathbf G_c$ instead of $\mathbf G_0$ is more appropriate for
practical applications with significant decoherence, as nonclassical
states are unlikely to survive in such an environment.

Consider a moment parameter of the form
\begin{align}
\beta(P) &= \int d^2\alpha P(\alpha) f(\alpha,\alpha^*),
\end{align}
where $f(\alpha,\alpha^*)$ is a real polynomial of $\alpha$ and
$\alpha^*$. For example, one may be interested in the mean of a
quadrature, in which case
$f = \alpha \exp(-i\theta) + \alpha^* \exp(i\theta)$, or the mean
energy, in which case $f = |\alpha|^2$.  The optical equivalence
theorem \cite{mandel} gives
\begin{align}
\beta &= \trace \rho Y,
&
Y &= :f(a,a^\dagger):,
\end{align}
where $:f(a,a^\dagger):$ denotes the normal ordering \cite{mandel}.
It follows from Sec.~\ref{sec_gradient} that an influence operator is
$\delta = Y-\beta$.

The next step is to find the tangent space of $\mathbf G_c$.  Although
$\mathbf G_c$ is a smaller family than $\mathbf G_0$, its dimension
turns out to be just as high.
\begin{proposition}
$\mathbf G_c$ is full-dimensional.
\label{classicalT}
\end{proposition}
\begin{proof}
Delegated to Appendix~\ref{sec_classicalT}.
\end{proof}
With the full-dimensional tangent space, the GHB is also given by
Eq.~(\ref{GHB_Y}). This result shows that the obvious von Neumann
measurement of $Y$ remains efficient in estimating $\beta$, and no
alternative measurements can do better, despite restricting the family
to classical states.  For example, if $f(\alpha,\alpha^*)$ is a
quadrature, then the homodyne measurement is efficient, and if
$f(\alpha,\alpha^*) = |\alpha|^2$, then
$:f(a,a^\dagger): = a^\dagger a$, and the photon-number measurement is
efficient.

$\mathcal G$, the space of positive densities, is
infinite-dimensional.  The estimation of $P$ would be a nonparametric
problem \cite{artiles05}, in contrast with the semiparametric problems
studied here. In classical statistics, it is known that a
nonparametric estimation of the probability density cannot achieve a
parametric convergence rate ($\mathsf{E} = O(1/N)$)
\cite{ibragimov81,bickel93,tsybakov}, and this difficulty is expected
to translate to the quantum domain.  Semiparametric estimation, on the
other hand, can achieve the parametric rate and is the more feasible
task if one is interested in only a few parameters of the system.

A further restriction on the family of $P$ can give very different
results, as shown in the next section in the context of incoherent
imaging.

\subsection{\label{sec_imaging}Incoherent imaging}

\subsubsection{The mother model}
Here we summarize existing results concerning the problem of
incoherent imaging \cite{tsang19a} using the language of
semiparametrics. Unlike previous sections, this section presents no
new results essentially. Rather, the goal is to use this very
important but equally difficult problem to illustrate the concepts and
current limitations of the quantum semiparametric theory.

\begin{figure}[htbp!]
\centerline{\includegraphics[width=0.45\textwidth]{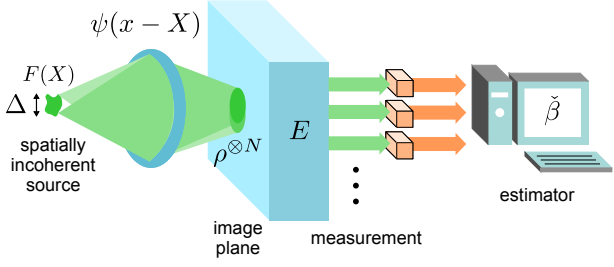}}
\caption{\label{spade}A far-field incoherent optical imaging system.}
\end{figure}

The basic setup of an imaging system is depicted in Fig.~\ref{spade}.
The object is assumed to emit spatially incoherent light at an optical
frequency. For simplicity, the imaging system is assumed to be
one-dimensional, paraxial, and diffraction-limited.  A model of each
photon on the image plane is \cite{tsang17,tsang19,tsang19a}
\begin{align}
\mathbf G &= \BK{\rho(F) = \int dX F(X)\ket{\psi_X}\bra{\psi_X}:
F \in \mathcal G_1},
\label{mother_imaging}
\\
\ket{\psi_X} &= \int dx \psi(x-X)\ket{x} = \exp(-ikX)\ket{\psi_0},
\end{align}
where $F$ is the unknown source density, $\mathcal G_1$ is a
  set of probability densities on $\mathbb R$, $X \in \mathbb R$ is
the object-plane coordinate, $\psi(x)$ is the point-spread function of
the imaging system, $x \in \mathbb R$ is the image-plane coordinate
normalized with respect to the magnification factor \cite{goodman},
$\ket{x}$ is the Dirac position ket that satisfies
$\braket{x|x'} = \delta(x-x')$, and $k$ is the canonical momentum
operator. $X$ and $x$ are further assumed to be normalized with
respect to the width of $\psi(x)$ so that they are dimensionless.
$\psi(x)$ is assumed here to be
\begin{align}
\psi(x) &= \frac{1}{(2\pi)^{1/4}} \exp\bk{-\frac{x^2}{4}},
\end{align}
such that $\ket{\psi_X} = \ket{\alpha = X/2}$ is a coherent state.
Various generalizations can be found in
Refs.~\cite{tsang17,tsang19,tsang19a,tsang19b} and references
therein. Besides imaging, the model can also be used to describe a
quantum particle under random displacements
\cite{hall09,*vidrighin,*branford19,ng16}.

The problem is semiparametric if $\mathcal G_1$ is
infinite-dimensional, such as
\begin{align}
\mathcal G_1
&= \textrm{all probability densities on $\mathbb R$},
\label{Gincoh}
\end{align}
and the parameter of interest is a functional of $F$, such as the
object moment
\begin{align}
\beta_\mu(F) &= \int   dX F(X) X^\mu,
\end{align}
where $\mu \in \mathbb N_1$ denotes the order of the moment of
interest. Notice that the family indicated by Eq.~(\ref{Gincoh}) is
much smaller than the one given by Eq.~(\ref{Gclassical}) in the
previous example, as the Glauber-Sudarshan function is now separable
in terms of $(\alpha',\alpha'')$ and confined to the real axis of
$\alpha$, viz.,
\begin{align}
P(\alpha) =2F(2\alpha')\delta(\alpha'').
\end{align}
In fact, the dimension of $\mathcal T^\perp$ is now infinite, as shown
in Appendix~\ref{sec_Tperp}, so this problem is the most difficult
type described in Sec.~\ref{sec_phil}.

The errors and their bounds are all functionals of the true density
$F$, and we will focus on their values for subdiffraction
distributions, which are defined as those with a width $\Delta$ around
$X = 0$ much smaller than the point-spread-function width, or in other
words $\Delta \ll 1$ \cite{tsang19a}.

\subsubsection{Semiparametric measurements and estimators}
Two globally unbiased measurements for semiparametric moment
estimation are known \cite{tsang19b}.  For $N$ detected photons
\footnote{ Refs.~\cite{tnl,tsang17,tsang19b} use the symbol $L$ for
  the number of detected photons, which is stochastic, and $N$ for the
  expected number of detected photons. For optics models with Poisson
  statistics, the conditioning on the detected photon number does not
  introduce any significant difference to the theory.}, both are
separable measurements and sample means in the form of \footnote{In
  practice, a histogram of the photon counts at the detectors provides
  sufficient statistics for the estimators and the photons do not need
  to be resolved individually \cite{tsang19b}.}
\begin{align}
E^{(N)}(\mathcal A_1,\mathcal A_1,\dots,
\mathcal A_N) &= \bigotimes_{n=1}^N E(\mathcal A_n),
\quad
\mathcal A_n \in \Sigma_{\mathcal X},
\\
\check\beta^{(N)}(\lambda_1,\lambda_2,\dots,\lambda_N) 
&= \frac{1}{N}\sum_{n=1}^N \check\beta(\lambda_n),
\quad
\lambda_n \in \mathcal X.
\end{align}
The first measurement is direct imaging, which measures the intensity
on the image plane and is equivalent to the projection of each photon
in the position basis as
\begin{align}
dE^{({\rm direct})}(x) &=  dx \ket{x}\bra{x},
\quad
x \in \mathcal X = \mathbb R.
\end{align}
An unbiased semiparametric estimator is given by the sample mean of
\begin{align}
\check\beta_\mu^{(\rm direct)}(x) &= \sum_{\nu=0}^\mu (C^{-1})_{\mu\nu} x^\nu,
\\
C_{\mu\nu} &= 
\iverson{\mu\ge\nu}\begin{pmatrix}\mu\\ \nu\end{pmatrix}
\int dx |\psi(x)|^2 x^{\mu-\nu},
\\
\iverson{{\rm proposition}} &= \begin{dcases}
1, & {\rm if\ proposition\ is\ true},\\
0, & {\rm otherwise,}
\end{dcases}
\end{align}
and the error is
\begin{align}
\mathsf{E}_\mu^{(\rm direct)} &= \frac{O(1)}{N},
\label{MSE_direct}
\end{align}
where $O(1)$ denotes a prefactor that does not scale with $\Delta$ in
the first order.  The second measurement is the so-called spatial-mode
demultiplexing or SPADE \cite{tnl,tsang17,tsang19,tsang19a,tsang19b},
which demultiplexes the image-plane light in the Hermite-Gaussian
basis given by
\begin{align}
\ket{\phi_m} &= \int dx \phi_m(x)\ket{x},
\quad
m \in \mathbb N_0,
\\
\phi_m(x) &= \frac{\operatorname{He}_m(x)}{(2\pi)^{1/4}\sqrt{m!}}
\exp\bk{-\frac{x^2}{4}},
\end{align}
where $\operatorname{He}_m(x)$ is a Hermite polynomial \cite{olver}.
For the estimation of an even moment with $\mu = 2j$, the POVM
for each photon is
\begin{align}
E^{({\rm SPADE})}(m) &= \ket{\phi_m}\bra{\phi_m},
\quad m \in \mathcal X = \mathbb N_0,
\end{align}
an unbiased semiparametric estimator is given by the sample mean of
\begin{align}
\check\beta_{2j}^{(\rm SPADE)}(m) &= \iverson{m\ge j} \frac{4^j m!}{(m-j)!},
\end{align}
and the error is 
\begin{align}
\mathsf{E}_{2j}^{({\rm SPADE})} &= \frac{O(\Delta^{2j})}{N}
=\frac{O(\Delta^{\mu})}{N},
\label{MSE_SPADE}
\end{align}
which is much lower than that of direct imaging in the subdiffraction
regime for the second and higher moments. For the estimation of odd
moments with SPADE, only approximate results have been obtained so far
\cite{tsang17,tsang18a,zhou19,bonsma19} and are not elaborated here.

Both estimators are efficient for their respective measurements in the
classical sense \cite{tsang19b}. In the quantum case, the question is
whether SPADE is efficient or there exist even better
measurements. Computing the GHB, or at least bounding it, would answer
the question and establish the fundamental quantum efficiency for
incoherent imaging.


\subsubsection{Lower bounds via parametric submodels}
Both Eqs.~(\ref{MSE_direct}) and (\ref{MSE_SPADE}) are upper bounds on
the GHB. By virtue of Proposition~\ref{loose}, all earlier quantum
lower bounds derived for incoherent imaging via parametric models are
in fact lower bounds on the GHB for the mother family given by
Eq.~(\ref{mother_imaging}), with the true $\rho$ being evaluated at
certain special cases of
$F$. References~\cite{tnl,bisketzi19,*lupo20a}, for example, assume
discrete point sources, but exact results become difficult to obtain
for a large number of sources. Here we highlight two methods that work
for any $F$ but can only give looser bounds.

The first method is the culmination of Ref.~\cite[Sec.~6]{tsang17} and
Ref.~\cite[Appendix~C]{tsang19}.  Assume that
\begin{align}
\theta &= \begin{pmatrix}\theta_g\\ \theta_h
\end{pmatrix}
\end{align}
consists of two sets of parameters
$\theta_g = (\theta_{g1},\theta_{g2},\dots)^\top$ and
$\theta_h = (\theta_{h0},\theta_{h1},\dots)^\top$.  Define a submodel
given by
\begin{align}
\sigma(\theta) &= \int dX F(X|\theta)\ket{\psi_X}\bra{\psi_X},
\\
F(X|\theta) &= \int dY \delta\bk{X-h(Y|\theta_h)} G(Y|\theta_g),
\\
G(Y|\theta_g) &= \frac{[1+\tanh g(Y|\theta_g)] F(Y)}
{\int dY [1+\tanh g(Y|\theta_g)] F(Y)}.
\end{align}
The truth is at 
\begin{align}
\sigma(0) &= \rho, &
F(X|0) &= G(X|0) = F(X),
\\
h(Y|0) &= Y, & g(Y|0) &= 0.
\end{align}
$\sigma(\theta)$ can be rewritten as
\begin{align}
\sigma(\theta) &= \int dY G(Y|\theta_g) 
\ket{\psi_{h(Y|\theta_h)}}\bra{\psi_{h(Y|\theta_h)}}.
\label{sigma_sub}
\end{align}
In other words, we have introduced parameters to both the mixing
density and the displacement in the model by rewriting the mixture.
Appendix~\ref{sec_EC} shows how the extended convexity of the Helstrom
information \cite{alipour,ng16} can be used on Eq.~(\ref{sigma_sub})
to give
\begin{align}
\tilde{\mathsf{H}}_\mu^{(N)} &\ge
\frac{\mathsf{H}_\mu^\sigma}{N} \ge 
\frac{\beta_{2\mu}-\beta_\mu^2 + \mu^2\beta_{2\mu-2}}{N}
= \frac{O(\Delta^{2\mu-2})}{N}.
\label{ECbound}
\end{align}
A more careful calculation shows that the SPADE error is exactly equal
to this bound for $\mu = 2$ \cite{tsang19b}. For higher moments,
however, Eq.~(\ref{ECbound}) remains much lower than that achievable
by SPADE.

The second method, as reported in Ref.~\cite{tsang19}, considers the
formal expansion $\exp(-ikX) = \sum_{p=0}^\infty (-ikX)^p/p!$, which
leads to
\begin{align}
\sigma(\theta) &= \sum_{p_1=0}^\infty \sum_{p_2=0}^\infty\beta_{p_1+ p_2}
\frac{(-ik)^{p_1}}{p_1!}\ket{\psi_0}\bra{\psi_0} \frac{(ik)^{p_2}}{p_2!}.
\end{align}
Consider this as a parametric submodel with only one scalar
parameter $\theta = \beta_\mu$ for a given $\mu$ , while all the
other moments $\beta_\nu$ with $\nu \neq \mu$
are fixed. Then the Helstrom bound for this
submodel is simply $\mathsf{H}_\mu^\sigma = 1/K^\sigma_{\mu\mu}$,
where $K^\sigma_{\mu\mu}$ is the Helstrom information
with respect to $\theta = \beta_\mu$. Reference~\cite{tsang19} finds
via a purification technique that this Helstrom bound is in turn bounded by
\begin{align}
\mathsf{H}_\mu^\sigma &=\frac{1}{K^\sigma_{\mu\mu}}
 \ge O(\Delta^{2\lfloor\mu/2\rfloor}).
\end{align}
By virtue of Corollary~\ref{iid2} and Proposition~\ref{loose}, 
we obtain
\begin{align}
\tilde{\mathsf{H}}_\mu^{(N)} &\ge 
\frac{\mathsf{H}_\mu^\sigma}{N} 
\ge \frac{O(\Delta^{2\lfloor\mu/2\rfloor})}{N}.
\end{align}
This lower bound does match the performance of SPADE in order of
magnitude, but it does not have a simple closed-form expression, and
the question of whether SPADE is exactly efficient for moments higher
than the second remains open.


\section{\label{sec_displacement}Semiparametric estimation with explicit
  nuisance parameters}

\subsection{The efficient score operator}
We now consider problems where there is an explicit partition of the
parameters into a scalar $\beta$ and nuisance parameters $\eta$ that
may be infinite-dimensional, viz.,
\begin{align}
\mathbf G &= \BK{\rho(\beta,\eta): \beta \in \Theta_\beta \subseteq \mathbb R, 
\eta \in \mathcal G}.
\end{align}
An example is the displacement model given by
Eq.~(\ref{displacement}), where $\beta$ is the displacement parameter
and the initial state $\rho_0$ depends on the nuisance parameters. All
previous studies of the problem assume that $\rho_0$ is known
exactly. In practice, however, $\rho_0$ may be poorly characterized,
and the estimation performance in the presence of unknown nuisance
parameters may suffer as a result.


With the explicit partition of the parameters, the scores can
  be partitioned similarly. Let $S^\beta$ be the score with respect to
  the parameter of interest, as defined by
\begin{align}
\parti{\rho(\beta,\eta)}{\beta} &= \rho \circ S^\beta,
\end{align}
where $\eta$ is fixed at the truth.
To define the nuisance scores, consider the subfamily
\begin{align}
\mathbf G_\eta &\equiv \BK{\rho(\beta,\eta): \eta \in \mathcal G},
\end{align}
which holds $\beta$ fixed at the truth instead.  Define the nuisance
tangent set $\{S^\eta\}$ as the set of scores from all parametric
submodels of $\mathbf G_\eta$ and the nuisance tangent space as
\begin{align}
\Lambda &\equiv \cspn\BK{S^\eta}.
\end{align}
The unbiasedness condition for an influence operator becomes
\begin{align}
\Avg{S^\beta,\delta} &= \parti{\beta}{\beta} = 1,
&
\Avg{S^\eta,\delta} &= 0.
\label{unbiased_eta}
\end{align}
The second of Eqs.~(\ref{unbiased_eta}) implies that
$\delta \perp \Lambda$, so if $S^\beta \in \Lambda$,
$\avg{S^\beta,\delta} = 0$, and no influence operator that obeys both
Eqs.~(\ref{unbiased_eta}) can exist.  In that case we assume the GHB
to be infinite. Provided that $S^\beta \notin \Lambda$, however, the
following theorem provides another method of computing the efficient
influence and the GHB.
\begin{theorem}
  Assuming $S^\beta \notin \Lambda$ and the unbiasedness condition
  given by Eqs.~(\ref{unbiased_eta}), the efficient influence and the
  GHB are given by
\begin{align}
\delta_{\rm eff} &= \frac{S_{\rm eff}}{\norm{S_{\rm eff}}^2},
&
\tilde{\mathsf{H}} &= \frac{1}{\norm{S_{\rm eff}}^2},
\end{align}
where $S_{\rm eff}$, henceforth called the efficient score, is given
by
\begin{align}
S_{\rm eff} &= S^\beta - \Pi(S^\beta|\Lambda).
\label{Seff}
\end{align}
\label{thm_Seff}
\end{theorem}
\begin{proof}
Delegated to Appendix~\ref{sec_Seff}.
\end{proof}
Figure~\ref{efficient_score} illustrates the Hilbert-space concepts
involved in Theorem~\ref{thm_Seff}. We note that
  Ref.~\cite[Sec.~5]{suzuki19} has also arrived at conclusions similar
  to Theorem~\ref{thm_Seff} in the parametric case, but the crucial
  point here is the Hilbert-space approach, which will enable us to
  derive closed-form solutions to semiparametric problems, as shown in
  the next section.

\begin{figure}[htbp!]
\centerline{\includegraphics[width=0.45\textwidth]{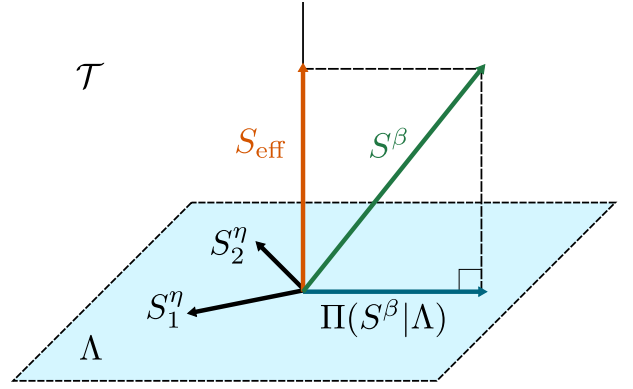}}
\caption{\label{efficient_score} The whole space in the picture
  represents the tangent space $\mathcal T$.  $\Lambda$ is the
  nuisance tangent space spanned by the nuisance tangent set
  $\{S^\eta\}$.  $S^\beta$ is the score with respect to the parameter
  of interest. The efficient score $S_{\rm eff}$ is $S^\beta$ minus
  its projection $\Pi(S^\beta|\Lambda)$. The result is orthogonal to
  $\Lambda$.}
\end{figure}

\subsection{Displacement estimation with a constrained 
family of initial states}

Consider the displacement model given by Eq.~(\ref{displacement}) and
illustrated by Fig.~\ref{displacement_example}.  For high-dimensional
systems, only a few moments of the initial state $\rho_0$ may be known
in practice, and it is prudent to assume that $\rho_0$ is in the
constrained family $\mathbf G_\gamma$ defined by
Eq.~(\ref{family_con}). The density-operator family for the problem
can be expressed as
\begin{align}
\mathbf G &= \BK{\rho(\beta,\rho_0)
= \mathcal U_\beta \rho_0: \beta \in \Theta_\beta \subseteq \mathbb R,
\rho_0 \in \mathbf G_\gamma},
\label{semi_displacement}
\end{align}
where the unitary map $\mathcal U_\beta$ is defined as
\begin{align}
\mathcal U_\beta\rho_0 &\equiv \exp\bk{-iH\beta} \rho_0 \exp\bk{iH\beta}.
\label{displace}
\end{align}
Generalization for more complicated generators is possible \cite{twc}
but outside the scope of this paper.

\begin{figure}[htbp!]
\centerline{\includegraphics[width=0.45\textwidth]{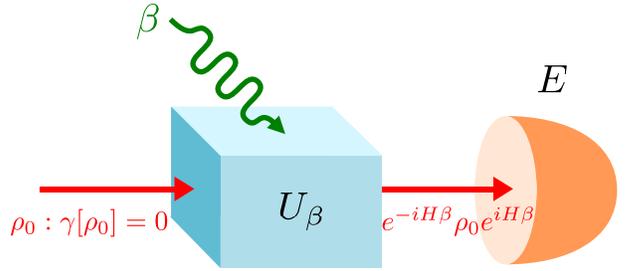}}
\caption{\label{displacement_example}A schematic of the semiparametric
  displacement model given by Eq.~(\ref{semi_displacement}).}
\end{figure}

Define an inner product and a norm with respect to the true $\rho_0$
as
\begin{align}
  \Avg{h_1,h_2}_0 &\equiv \trace \rho_0 h_1 \circ h_2,
&
    \norm{h}_0 &\equiv \sqrt{\Avg{h,h}_0}.
\end{align}
Define also the operator Hilbert space $\mathcal Z_0$ with
  respect to $\rho_0$, the tangent space $\mathcal T_0$ at $\rho_0$
  with respect to $\mathbf G_\gamma$, and the orthocomplement
  $\mathcal T_0^\perp$ that gives
  $\mathcal Z_0 = \mathcal T_0 \oplus \mathcal T_0^\perp$, in the same
  way as how the spaces $\mathcal Z$, $\mathcal T$, and
  $\mathcal T^\perp$ are defined with respect to $\rho$.  Noting the
  unitarity of $\mathcal U_\beta$ and following the method in
  Appendix~\ref{sec_iid2}, it can be shown that the nuisance tangent
  space is given by
\begin{align}
\Lambda &= \mathcal U_\beta\mathcal T_0 \equiv 
\BK{\mathcal U_\beta h: h \in \mathcal T_0}.
\label{semitangent2}
\end{align}
Define the map adjoint to $\mathcal U_\beta$ by
$\mathcal U_\beta^*h \equiv \exp(iH\beta) h \exp(-iH\beta)$.  Exploiting the
isomorphism between $\Lambda$ and $\mathcal T_0$, we can compute the
efficient score as follows:
\begin{align}
S_{\rm eff} &= S^\beta - \Pi(S^\beta|\Lambda)
= S^\beta - \mathcal U_\beta\Pi(\mathcal U_\beta^*S^\beta|\mathcal T_0)
\\
&= S^\beta - \mathcal U_\beta\Bk{\mathcal U_\beta^*S^\beta - \Pi(\mathcal U_\beta^*S^\beta|\mathcal T_0^\perp)}
\\
&= \mathcal U_\beta\Pi(\mathcal U_\beta^*S^\beta|\mathcal T_0^\perp)
\\
&= 
\Avg{R,\mathcal U_\beta^* S^\beta}_0^\top \Avg{R,R}_0^{-1}
\mathcal U_\beta R,
\label{Seff_R1}
\end{align}
where $R$ is the vector of antiscores with respect to $\rho_0$, as
defined by Eq.~(\ref{Rgradient}) but with $\rho_0$ and
$\avg{\cdot,\cdot}_0$ instead. 
Equation~(\ref{Seff_R1}) can be further simplified, with
\begin{align}
\Avg{R,\mathcal U_\beta^* S^\beta}_0 &= \Avg{\mathcal U_\beta R,S^\beta} 
\\
&= -i \trace \rho \Bk{\mathcal U_\beta R,H}
\label{comm_op}
\\
&= -i \trace \rho_0\Bk{R,H} = \Bk{R,H}_0,
\end{align}
where $[A,B]_{jk} \equiv A_jB_k-B_kA_j$ and $[\cdot,\cdot]_0$ is
shorthand for $-i\trace \rho_0[\cdot,\cdot]$. Equation~(\ref{comm_op})
comes from the fact that $S^\beta = \mathfrak D H$ for the
  model given by Eq.~(\ref{displace}), where $\mathfrak D$ is the
  so-called commutation superoperator defined by
  \cite{holevo11,holevo77}
\begin{align}
\Avg{h,\mathfrak D H} &= -i\trace \rho\Bk{h,H}
\quad
\forall h \in \mathcal Y.
\label{commutation}
\end{align}
The final result is
\begin{align}
S_{\rm eff} &= \Bk{R,H}_0^\top\Avg{R,R}_0^{-1}\mathcal U_\beta R,
\label{Seff_R2}
\\
\norm{S_{\rm eff}}^2 &=
\Bk{R,H}_0^\top \Avg{R,R}_0^{-1}\Bk{R,H}_0,
\label{Seff_norm}
\\
\tilde{\mathsf{H}}^{(N)} &= \frac{1}{N\norm{S_{\rm eff}}^2}.
\label{GHB_semi}
\end{align}
In particular, if the constraint is linear and a scalar given by
\begin{align}
\trace \rho_0 Z &= 0,
&
R &= Z,
\end{align}
then
\begin{align}
\tilde{\mathsf{H}}^{(N)} &=  \frac{\norm{Z}_0^2}{N[Z,H]_0^2},
\label{GHB_semi2}
\end{align}
which gives Eq.~(\ref{displacement_bound}). $\norm{Z}_0^2$ is the variance of
$Z$, while
\begin{align}
[Z,H]_0 = -i\trace \rho_0\Bk{Z,H} = 
\trace \rho_0 \left.\parti{}{\beta} \mathcal U_\beta^*Z\right|_{\beta=0}
\end{align}
is a measure of how sensitive the Heisenberg-picture $Z$ is to the
displacement. An intuitive explanation of this result is as follows. A displacement
can be estimated only with respect to a known reference.  If only the
mean of $Z$ is known about the initial state, then it is the only
reference in the quantum object that is available to the observer. It
is therefore not surprising---in hindsight---that the statistics of
$Z$ determine the fundamental limit.

If $H$ is the momentum operator and $Z$ is the position operator
satisfying $[Z,H] = i$, the Heisenberg picture of $Z$ is
\begin{align}
\mathcal U_\beta^* Z &= \beta + Z,
\end{align}
which is a quantum additive-noise model with no known statistics about
the noise operator $Z$ other than its mean.  Measurements of $Z$ and
the sample mean of the outcomes are efficient. This problem then
becomes equivalent to the $\beta = \trace\rho Y$ example, but note
that Eqs.~(\ref{Seff_norm}) and (\ref{GHB_semi}) are more general, as
they can deal with any generator, a $\beta$ that cannot be easily
expressed as a functional of $\rho$, and more general constraints.

Another example is optical phase estimation with
\begin{align}
H = a^\dagger a = \frac{1}{2}\bk{Z_1^2+Z_2^2},
\end{align}
and constraint $\trace \rho_0 Z = \zeta$ on the mean of the quadrature
operators $Z = (Z_1,Z_2)^\top$ with $[Z_1,Z_2] = i$. There is no phase
observable \cite{mandel}, so expressing $\beta$ as a functional of
$\rho$ is difficult if not impossible. Equations~(\ref{Seff_norm}) and
(\ref{GHB_semi}), on the other hand, are simple expressions in terms
of the generator and the antiscores.  In
Eqs.~(\ref{Seff_R2})--(\ref{GHB_semi}), $R = Z - \zeta$,
\begin{align}
\Avg{R_j,R_k}_0 = \trace \rho_0 \bk{Z_j-\zeta_j}\circ \bk{Z_k-\zeta_k}
\end{align}
is simply the covariance matrix of the quadratures, while
\begin{align}
\Bk{R_1,H}_0 &= \Bk{Z_1,H}_0 = \trace \rho_0Z_2 = \zeta_2, 
\\
\Bk{R_2,H}_0 &= \Bk{Z_2,H}_0 = -\trace \rho_0 Z_1 = -\zeta_1
\end{align}
are the mean quadrature values. The efficient influence
$\delta_{\rm eff}\propto S_{\rm eff}$ is a linear combination of the
quadratures according to Eq.~(\ref{Seff_R2}), indicating the ideal,
though parameter-dependent, quadrature to be measured.  An adaptive
measurement can then aim to measure the ideal quadrature to approach
the quantum limit.

When $\rho_0$ is exactly known, the Helstrom bound for displacement
estimation has been computed exactly only if $\rho_0$ is pure or
Gaussian. Only looser bounds have been found otherwise
\cite{helstrom,holevo11,demkowicz15}.  The Mandelstam-Tamm inequality,
for example, is looser than the Helstrom bound for mixed states
\cite{holevo11}. $S^\beta$ is determined by $\mathfrak D H$, and if
$\rho_0$ is a high-dimensional non-Gaussian mixed state, $S^\beta$ is
intractable. With the infinitely many nuisance parameters and
infinitely many scores assumed here, the problem is hopeless under the
conventional bottom-up approach. The top-down geometric approach, on
the other hand, is able to avoid the computation of the scores
altogether and give a simple result in terms of the more tractable
antiscores.

\section{\label{sec_multi}Vectoral parameter of interest}
To complete the formalism, here we generalize the core results in this
paper for a vectoral parameter of interest $\beta \in \mathbb R^q$
with $q \ge 1$ entries. $p$, the dimension of the parameter space,
should be at least as large as $q$ and may be infinite. Define the
error matrix as
\begin{align}
\Sigma &\equiv \int \Bk{\check\beta(\lambda)-\beta} 
\Bk{\check\beta(\lambda)-\beta}^\top \trace dE(\lambda)\rho,
\end{align}
where $\check\beta:\mathcal X \to \mathbb R^q$ is an estimator.  An
influence operator should then be a vector of $q$ operators. The
inner product between two vectoral operators and the norm are now
defined as
\begin{align}
\trace \Avg{h,g} &= \sum_{j=1}^q \trace \rho\bk{h_j\circ g_j},
&
\norm{h} &\equiv \sqrt{\trace\Avg{h,h}}.
\end{align}
The Hilbert spaces $\mathcal Y$ and $\mathcal Z$ for the vectoral
operators are still expressed as Eqs.~(\ref{Y}) and (\ref{Z}), while
the tangent space is now defined as the replicating space
\cite{tsiatis06}
\begin{align}
\mathcal T &\equiv \bk{\cspn\BK{S}}^{\oplus q} 
\equiv \underbrace{\cspn\BK{S} \oplus \dots \oplus \cspn\BK{S}}
_{q \textrm{ terms}}.
\end{align}
The set of influence operators is still given by
Eq.~(\ref{influence_set}) if $\avg{S,\delta} = \partial\beta$ is
interpreted as $\avg{S^\sigma,\delta_k} = \partial\beta_k$ for all
submodels and $k = 1,\dots,q$.  For an unbiased measurement, the error
operator given by Eq.~(\ref{error}) remains an element of
$\mathcal D$, and it can be shown \cite[Sec.~6.2]{holevo11} that
\begin{align}
\Sigma &\ge \Avg{\delta,\delta},
\end{align}
where the matrix inequality $A \ge B$ means that $A-B$ is
positive-semidefinite. The GHB can then be expressed as
\begin{align}
\mathsf{E} &\equiv \trace W \Sigma 
\ge \trace W \Avg{\delta,\delta}
\ge \inf_{\delta \in\mathcal D} \trace W \Avg{\delta,\delta}
\equiv \tilde{\mathsf{H}},
\end{align}
where $W \ge 0$ is a real cost matrix
\cite{hayashi}. Generalizing Theorems~\ref{thm_GHB} and
\ref{thm_SHB}, we have
\begin{theorem}
\label{thm_GHB_vec}
The GHB for a vectoral parameter of interest is given by
\begin{align}
\tilde{\mathsf{H}} &=
\min_{\delta \in\mathcal D} \trace W \Avg{\delta,\delta}
= \trace W \Avg{\delta_{\rm eff},\delta_{\rm eff}},
\label{GHB_vec}
\end{align}
where the efficient influence $\delta_{\rm eff}$ is the unique element
in $\mathcal D$ given by
\begin{align}
\delta_{\rm eff} &= \Pi(\delta|\mathcal T).
\end{align}
\end{theorem}
\begin{proof}
  Delegated to Appendix~\ref{sec_GHB_vec}.
\end{proof}
It is straightforward to generalize the methods introduced in this
paper to compute the GHB for the vectoral case.

Holevo proposed another bound, denoted in the following by the
sans-serif $\mathsf X$, that can account for the quantum effect of
observable incompatibility in multiparameter estimation
\cite{holevo11,nagaoka89}. \PRXagain{Before we prove the bound and
  related results, we need the following lemma.

\begin{lemma}[Belavkin and Grishanin \cite{belavkin73}]
For any complex positive-semidefinite matrix $A$,
\begin{align}
\trace \real A &\ge \norm{\imag A}_1,
\end{align}
where $\real A$ and $\imag A$ denote the entry-wise real and imaginary
parts of $A$, respectively, and $\norm{\cdot}_1$ denotes the trace
norm, defined as the sum of the singular values.
\label{lem_belavkin}
\end{lemma}
\begin{proof}
Provided in Appendix~\ref{sec_belavkin} for completeness.
\end{proof}

We can now present the Holevo bound. It requires little modification
to be applied to semiparametric estimation; only the definition of
$\mathcal D$ needs to be generalized to Eq.~(\ref{influence_set})
here. Otherwise the proof is standard
\cite{holevo11,nagaoka89,demkowicz20}; we provide it here simply to
demonstrate that it remains valid in the semiparametric setting.  }
\begin{theorem}
\label{thm_holevo}
\begin{align}
\mathsf{E} &\ge \mathsf{X} \equiv \inf_{\delta\in\mathcal D}
\Bk{\trace W \real \Gamma(\delta) + 
\norm{\sqrt{W}\imag \Gamma(\delta)\sqrt{W}}_1},
\label{X}
\end{align}
where $\Gamma(\delta)$ is a complex matrix given by
\begin{align}
\Gamma_{jk}(\delta) \equiv \trace \rho \delta_j\delta_k.
\end{align}
\end{theorem}
\PRXagain{
\begin{proof}
Holevo proved \cite[Eq.~(6.6.55)]{holevo11} that the error matrix and
the error operator of any unbiased measurement obeys
\begin{align}
\Sigma \ge \Gamma(\delta).
\end{align}
Thus $A = \sqrt{W}(\Sigma-\Gamma)\sqrt{W} \ge 0$.  Applying
Lemma~\ref{lem_belavkin} and noting that $\Sigma$ is real, we obtain
\begin{align}
\trace \real A = 
\trace \sqrt{W}(\Sigma-\real\Gamma)\sqrt{W}
= \trace W(\Sigma-\real\Gamma) 
\nonumber\\
\ge \norm{\imag\Bk{\sqrt{W}(\Sigma-\Gamma)\sqrt{W}}}_1
= \norm{\sqrt{W}\imag \Gamma\sqrt{W}}_1.
\end{align}
Hence
\begin{align}
\trace W \Sigma &\ge \trace W \real\Gamma
+\norm{\sqrt{W}\imag \Gamma\sqrt{W}}_1
\ge \mathsf X.
\end{align}
\end{proof}
}

The asymptotic attainability of the Holevo bound for $d < \infty$ has
been shown in Refs.~\cite{kahn09,gill_guta,demkowicz20}.  The rough
idea there is to consider a two-step method: first find an estimate
$\check\theta$ of $\theta$ using some of the object copies, and then
perform a measurement based on the influence operators obtained from
the minimization in Eq.~(\ref{X}), assuming $\check\theta$ to be the
truth. In the limit of $N\to\infty$, the overhead for finding
$\check\theta$ is benign, and it can be shown that the error
approaches $\mathsf X$ by local asymptotic normality.

For all the examples studied in previous sections, $\beta$ was a
scalar, and it is straightforward to prove that the Holevo bound is
equal to the GHB in that case.
\begin{corollary} If $\beta$ is a scalar ($q = 1$),
\label{cor_holevo_scalar}
\begin{align}
\mathsf X = \tilde{\mathsf H}.
\end{align}
\end{corollary}
\begin{proof}
  For $q = 1$, $\Gamma(\delta) = \trace \rho \delta^2$ and
  $\imag \Gamma(\delta) = 0$, leading to
\begin{align}
\mathsf X &= \inf_{\delta\in\mathcal D} \trace W \real \Gamma(\delta)
=\inf_{\delta\in\mathcal D} \trace W \Avg{\delta,\delta} = \tilde{\mathsf H}.
\end{align}
\end{proof}
The scalar GHB hence inherits all the properties of the Holevo bound,
including its asymptotic attainability. In fact, for any $q$, the
Holevo bound turns out to be a marginal improvement over the GHB only.
\begin{theorem}
\begin{align}
\tilde{\mathsf{H}} \le \mathsf{X}  \le 2\tilde{\mathsf{H}}.
\label{sandwich}
\end{align}
\label{thm_sandwich}
\end{theorem}
\PRXagain{
\begin{proof}
For all $\delta \in \mathcal D$,
\begin{align}
\trace W \real \Gamma(\delta) + \norm{\sqrt{W}\imag \Gamma(\delta)\sqrt{W}}_1
\label{holevo_expression}
\\
\ge \trace W \real \Gamma(\delta)
= \trace W \Avg{\delta,\delta} \ge \tilde{\mathsf H}.
\end{align}
As $\mathsf X$ is the infimum of Eq.~(\ref{holevo_expression}), we
obtain $\mathsf X \ge \tilde{\mathsf H}$, the first inequality of the
theorem. The second inequality is proved as follows:
\begin{align}
\mathsf{X} &\le 
\trace W \real \Gamma(\delta_{\rm eff}) 
+ \norm{\sqrt{W}\imag \Gamma(\delta_{\rm eff})\sqrt{W}}_1
\equiv \mathsf D
\label{D_invariant}
\\
&\le 
\trace W \real \Gamma(\delta_{\rm eff}) + \trace \sqrt{W} \real \Gamma(\delta_{\rm eff})
\sqrt{W}
\label{belavkin_step}
\\
&= 2 \trace W \real \Gamma(\delta_{\rm eff}) = 2\tilde{\mathsf H},
\end{align}
where Eq.~(\ref{belavkin_step}) is obtained by applying
Lemma~\ref{lem_belavkin}  to
$A = \sqrt{W} \Gamma(\delta_{\rm eff})\sqrt{W}$.
\end{proof}
}
The first inequality $\tilde{\mathsf H} \le \mathsf X$ is well known
\cite{holevo11,nagaoka89,ragy16}. A special case
$\mathsf X \le 2\mathsf H$ of the second inequality---when $p<\infty$,
$K^{-1}$ exists, and $\tilde{\mathsf H} = \mathsf H$ is the original
Helstrom bound---was proved recently in
Ref.~\cite{carollo19,*carollo20}. $\mathsf X = 2\mathsf H$ can be
attained in special cases \cite{kahn09,gill_guta,demkowicz20}. 

Theorem~\ref{thm_sandwich} implies that the effect of incompatibility
is surprisingly benign in the context of asymptotic statistics, the
GHB can be approached to within a factor of two if the Holevo bound is
attainable, and the GHB is a serviceable alternative to the Holevo
bound, especially when the latter is more difficult to compute.  See
Ref.~\cite{demkowicz20} for further interesting discussions regarding
this result.

\PRXagain{As an aside, we remark that the $\mathsf D$ in
  Eq.~(\ref{D_invariant}) is called the $\mathfrak D$-invariant bound
  and coincides with $\mathsf X$ if
  $\mathcal T = \mathfrak D \mathcal T$, where $\mathfrak D$ is given
  by Eq.~(\ref{commutation}) \cite{holevo11,suzuki16,*suzuki19a}. In
  general, $\mathsf D$ offers a tighter upper bound on $\mathsf X$
  than $2\tilde{\mathsf H}$ but may not be much more difficult to
  compute, as it also depends on $\delta_{\rm eff}$, which can be
  found via the methods introduced in this work.

We present a few other interesting results concerning multiparameter
estimation with $p < \infty$ in Appendix~\ref{sec_parametric}.
}

Finally, we generalize the concept of efficient score in Theorem~\ref{thm_Seff}
for a vectoral $\beta$.
\begin{theorem}
  Assume a density-operator family given by
  \begin{align}
\mathbf G &= \BK{\rho(\beta,\eta):\beta\in\Theta_\beta \subseteq \mathbb R^q, 
\eta \in\mathcal G}.
\end{align}
Let $S^\beta = (S^\beta_1,\dots,S^\beta_q)^\top$ be the scores with
respect to $\beta$ and $\{S^\eta\}$ be the nuisance tangent set.
Assume the unbiasedness condition for influence operators
$\delta\in\mathcal D$ given by
\begin{align}
\Avg{S^\beta,\delta} &= I, & \Avg{S^\eta,\delta} &= 0,
\end{align}
where $I$ is the identity matrix. The efficient influence and the GHB
are given by
\begin{align}
\delta_{\rm eff} &= \Avg{S_{\rm eff},S_{\rm eff}}^{-1} S_{\rm eff},
&
\tilde{\mathsf{H}} &= \trace W \Avg{S_{\rm eff},S_{\rm eff}}^{-1},
\end{align}
where the efficient score $S_{\rm eff}$ is given by
\begin{align}
S_{\rm eff} &= S^\beta - \Pi(S^\beta|\Lambda),
&
\Lambda &\equiv \bk{\cspn\BK{S^\eta}}^{\oplus q},
\end{align}
and $\avg{S_{\rm eff},S_{\rm eff}}^{-1}$ is assumed to exist.
\end{theorem}
\begin{proof}
  Almost identical to that of Theorem~\ref{thm_Seff} in
  Appendix~\ref{sec_Seff} and omitted here for brevity.
\end{proof}

\section{Conclusion}
We have founded a theory of quantum semiparametric estimation and
showcased its power by producing simple quantum bounds for a large
class of problems with high dimensions and few assumptions about the
density operator. The theory establishes the notion of quantum
semiparametric efficiency, which should inform and inspire the design
of more efficient measurements in many areas of quantum physics.

While the experimental design of efficient semiparametric measurements
is only touched upon here and awaits further research, the importance
of the quantum limits set forth should not be underestimated. As more
experiments are now being performed on complex quantum systems and
advantages of such systems for metrology and information processing in
general are being claimed, the precision limits serve as ultimate
yardsticks as well as ``no-go'' theorems that guard against spurious
proposals and fruitless endeavors, in the same way the laws of
thermodynamics impose limits to engines and rule out perpetual-motion
machines. Deriving precision limits for highly complex or poorly
modeled quantum systems was a daunting task under the curse of
dimensionality; the semiparametric theory offers a new way forward.

Many open problems still remain. More extensions and applications of
the theory remain to be worked out.  The asymptotic attainability of
efficiency \cite{hayashi,hayashi05,kahn09,gill_guta,demkowicz20} is a
thorny issue for infinite-dimensional problems.  The assumption of
unbiased estimation is a drawback; generalizations to the Bayesian or
minimax paradigm
\cite{bell,*schutzenberger57,*vantrees,*gill95,*personick71,*hayashi11,*liu16,*chabuda16,*rubio19,*rubio20}
should help but await further research.  These problems should benefit
from studies of alternative quantum bounds beyond the Cram\'er-Rao
type \cite{tsuda05,*glm2012,*qzzb,*qbzzb,*qwwb,*hall_prx,*nair18}. In
view of Eq.~(\ref{GHB_entropy}) and Figs.~\ref{manifolds} and
\ref{straight_lines}, the connections of quantum semiparametrics to
other domains of quantum information
\cite{tomamichel13,*li14a,*tomamichel16} and quantum state geometry
\cite{hayashi,hayashi05,amari} are also interesting future directions.

In light of the richness and wide applications of the classical
semiparametric theory
\cite{ibragimov81,bickel93,tsiatis06,newey90,feigelson12,tsang19b},
this work has only scratched the surface of the full potential of
quantum semiparametrics. It should open doors to further useful
results.

\begin{acknowledgments}
  We thank M.~G.~Genoni both for several fruitful discussions and for
  making us aware of Refs.~\cite{carollo19,*carollo20}. We are
  grateful to R.~Nair, M.~Gu\c{t}\u{a}, R.~Gill, D.~Branford,
  R.~Demkowicz-Dobra\'nzski, J.~F.~Friel, W.~G\'orecki, and J.~Suzuki
  for useful discussions. This research is partly supported by the
  National Research Foundation (NRF) Singapore, under its Quantum
  Engineering Programme (Award QEP-P7).  AD and FA have been supported
  by the UK EPSRC (EP/K04057X/2) and the UK National Quantum
  Technologies Programme (EP/M01326X/1, EP/M013243/1). FA also
  acknowledges financial support from the National Science Center
  (Poland) grant No. 2016/22/E/ST2/00559.
\end{acknowledgments}

\appendix

\section{\label{sec_HB} Proof of Corollary~\ref{cor_HB}}
If $p < \infty$ and $K^{-1}$ exists, the solution to
$\Pi(\delta|\mathcal T)$ can be found, for example, in
Ref.~\cite[Eq.~(15) in Appendix~A.2]{bickel93}. Here we give a simple
proof for completeness.  By definition of the projection
\cite{debnath05},
\begin{align}
\Pi(\delta|\mathcal T) &= \argmin_{h\in\mathcal T}\norm{\delta-h}.
\end{align}
Any $h \in \mathcal T$ can be expressed as the linear combination
$w^\top S$ with respect to a certain vector $w \in \mathbb R^p$. Then
\begin{align}
\norm{\delta-h}^2
&= \Avg{\delta,\delta} - w^\top \Avg{S,\delta}-
\Avg{S,\delta}^\top w + w^\top \Avg{S,S}w.
\end{align}
The solution to the least-squares problem is
\begin{align}
w_{\rm min} &= \Avg{S,S}^{-1}\Avg{S,\delta},
\\
\Pi(\delta|\mathcal T) &= w_{\rm min}^\top S
= \Avg{S,\delta}^\top \Avg{S,S}^{-1}S.
\label{delta_eff_formula}
\end{align}
Hence
\begin{align}
\norm{\Pi(\delta|\mathcal T)}^2 &= 
\Avg{S,\delta}^\top \Avg{S,S}^{-1}\Avg{S,\delta},
\label{delta_eff_norm}
\end{align}
which is equal to Eq.~(\ref{HB}), since
$\avg{S,\delta} = \partial\beta$ for an influence operator.

\section{\label{sec_iid}Proof of Corollary~\ref{iid}}
Denote any concept discussed so far with the superscript $(N)$ if it
is associated with $\mathbf F^{(N)}$, but omit the superscript $(1)$
for brevity if $N = 1$.  From $\mathcal Z$, we generate a subspace
$U\mathcal Z \subset \mathcal Z^{(N)}$ such that
  \begin{align}
U\mathcal Z \equiv \BK{Uh: h \in \mathcal Z}.
\end{align}
$U$ is a surjective map to $U\mathcal Z$ by definition of the space.
It can be shown that
  \begin{align}
\Avg{Uh_1,Uh_2}^{(N)} = \Avg{h_1,h_2}
\quad \forall h_1,h_2 \in \mathcal Z,
\end{align}
so $U\mathcal Z$ is isomorphic to $\mathcal Z$, and $U$ is a unitary
map from $\mathcal Z$ to $U\mathcal Z$ \cite{reed_simon}. It can also
be shown that
\begin{align}
S^{(N)} = \sqrt{N} U S,
\end{align}
so $\mathcal T^{(N)}= \cspn\{S^{(N)}\} \subseteq U\mathcal Z$,
and $\mathcal T^{(N)}$ is isomorphic to $\mathcal T$.  For any
$Uh \in U\mathcal Z$, it is not difficult to prove that
\begin{align}
\Pi(Uh|\mathcal T^{(N)}) = U\Pi(h|\mathcal T),
\end{align}
given the isomorphisms. Now let
\begin{align}
\delta^{(N)} = \frac{U\delta}{\sqrt{N}} \in U\mathcal Z,
\end{align}
where $\delta$ is an influence operator. $\delta^{(N)}$ is also an
influence operator, since
\begin{align}
\Avg{S^{(N)},\delta^{(N)}}^{(N)} = \Avg{\sqrt{N} U
  S,\frac{U\delta}{\sqrt{N}}}^{(N)} = \Avg{S,\delta} = \partial\beta.
\end{align}
The efficient influence for $\mathbf F^{(N)}$ becomes
\begin{align}
\delta_{\rm eff}^{(N)} &= \Pi(\delta^{(N)}|\mathcal T^{(N)}) = 
\frac{\Pi(U\delta|\mathcal T^{(N)})}{\sqrt{N}}
= 
\frac{U \Pi(\delta|\mathcal T)}{\sqrt{N}}
\nonumber\\
&= \frac{ U \delta_{\rm eff}}{\sqrt{N}},
\end{align}
the norm becomes
\begin{align}
\norm{\delta_{\rm eff}^{(N)}}^{(N)} = \frac{\norm{\delta_{\rm eff}}}{\sqrt{N}},
\end{align}
and the corollary ensues.

\section{\label{sec_iid2}Proof of Corollary~\ref{iid2}}
  Let $\{S^{(N)}\}$ be the tangent set for $\mathbf G^{(N)}$.  For
  each parametric submodel $\{\sigma(\theta)\}$ of $\mathbf G$, let
\begin{align}
\BK{\tau(\theta) = \sigma(\theta)^{\otimes N}}
\label{FsigmaN}
\end{align}
be a parametric submodel of $\mathbf G^{(N)}$. The score of the
submodel is given by
\begin{align}
S^\tau &= \sqrt{N} U S^\sigma.
\label{Stau}
\end{align}
In other words, each $S^\sigma \in \{S\}$ can be used to
generate a score in $\{S^{(N)}\}$ via Eq.~(\ref{Stau}).  The set of
scores generated this way is therefore a subset of $\{S^{(N)}\}$,
viz.,
\begin{align}
\BK{\sqrt{N}US}
&\equiv \BK{\sqrt{N}US^\sigma: S^\sigma \in \BK{S}}
\subseteq \BK{S^{(N)}}.
\end{align}
Conversely, any parametric submodel of $\mathbf G^{(N)}$ must
  be in the form of Eq.~(\ref{FsigmaN}), with $\{\sigma(\theta)\}$
  being a certain parametric submodel of $\mathbf G$.  The score of
  the former is then related to the score of the latter via
  Eq.~(\ref{Stau}).  Since $\{S\}$ includes the scores of all
  parametric submodels of $\mathbf G$, any $S^\tau \in \{S^{(N)}\}$
  must be in $\{\sqrt{N}US\}$. Thus $\{S^{(N)}\} \subseteq \{\sqrt{N}US\}$,
  and equality holds, viz., 
\begin{align}
\BK{S^{(N)}} &= \BK{\sqrt{N}US}.
\end{align}
It follows that
\begin{align}
\mathcal T^{(N)} \equiv \cspn\BK{S^{(N)}} = \cspn\BK{\sqrt{N}US}
\end{align}
is isomorphic to $\mathcal T = \cspn\{S\}$.  Hence, projecting an
influence operator of the form $\delta^{(N)} = U\delta/\sqrt{N}$ into
$\mathcal T^{(N)}$ gives the efficient influence
$\delta_{\rm eff}^{(N)} = U \delta_{\rm eff}/\sqrt{N}$, by the same
argument as Appendix~\ref{sec_iid}.

\section{\label{sec_cauchy} The set of bounded operators 
 is dense in $\mathcal Z$}
To generalize Theorem~\ref{TeqZ} for the infinite-dimensional case and
prove Theorem~\ref{TeqZ2}, we need to be mindful of the unbounded
operators in $\mathcal Z$.  The good news is that they are well
defined as limits of bounded-operator sequences in $\mathcal Y$,
thanks to Holevo \cite{holevo11,holevo77}; just a minor modification
is needed to make his result work for $\mathcal Z$.

Consider the set $\mathcal B$ of bounded elements defined by
Eq.~(\ref{set_bounded}). If $d < \infty$,
$\mathcal B = \overline{\mathcal B} = \mathcal Z$, since all operators
are bounded in the finite-dimensional case, but if $d = \infty$,
$\mathcal B \subset \mathcal Z$ is a strict subset. A useful lemma is
as follows.

\begin{lemma}
\label{cauchy}
$\overline{\mathcal B} = \mathcal Z$.
\end{lemma}
\begin{proof}
  Reference~\cite[Theorem~2.8.1]{holevo11} implies that, for any
  $h \in \mathcal Z \subset \mathcal Y$, there exists a Cauchy
  sequence $\{h_n\}$ with each $h_n\in \mathcal Y$ satisfying
  $\opnorm{h_n} < \infty$ such that
\begin{align}
\lim_{n\to \infty} \norm{h - h_n} &= 0.
\label{cauchy_y}
\end{align}
To derive a similar convergent sequence in $\mathcal Z$, consider the
projection of each $h_n$ into $\mathcal Z$, written as
\begin{align}
h_n' = \Pi(h_n|\mathcal Z) = h_n - \Avg{h_n,I} \in \mathcal Z.
\end{align}
Denote a bounded operator in the equivalence class of $h_n$ as
$\hat h_n$. An operator for $h_n'$ can be expressed as
\begin{align}
\hat h_n' = \hat h_n - \Avg{h_n,I}\hat I.
\end{align}
Since $\opnorm{\hat h_n} < \infty$ and
$\opnorm{\avg{h_n,I}\hat I} = |\avg{h_n,I}|< \infty$, 
\begin{align}
\opnorm{\hat h_n'} &\le \opnorm{\hat h_n} + 
\opnorm{\Avg{h_n,I}\hat I} < \infty
\end{align}
by the triangle inequality, leading to $h_n' \in \mathcal B$. The
Pythagorean theorem leads to
\begin{align}
\norm{h_n-h_m} &\ge \norm{h_n'-h_m'}
\quad
\forall n,m,
\\
\norm{h - h_n} &\ge \norm{h - h_n'},
\end{align}
which can be combined with Eq.~(\ref{cauchy_y}) to give
\begin{align}
\lim_{n\to\infty} \norm{h- h_n'} &= 0.
\end{align}
In other words, $\{h_n'\}$, with each $h_n' \in \mathcal B$, is also
Cauchy and converges to $h$. As the argument applies to any
$h \in \mathcal Z$, $\mathcal B$ is dense in $\mathcal Z$, and the
closure of $\mathcal B$ gives $\mathcal Z$.
\end{proof}

\section{\label{sec_loose}Proof of Proposition~\ref{loose}}
Let the orthocomplement of $\mathcal V$ in $\mathcal T$ be
$\mathcal V_{\mathcal T}^\perp$. Then the Pythogorean theorem yields
\begin{align}
\norm{\delta_{\rm eff}}^2 &= \norm{\Pi(\delta_{\rm eff}|\mathcal V)}^2
+\norm{\Pi(\delta_{\rm eff}|\mathcal V_{\mathcal T}^\perp)}^2
\\
&\ge \norm{\Pi(\delta_{\rm eff}|\mathcal V)}^2
= \norm{\Pi(\Pi(\delta|\mathcal T)|\mathcal V)}^2
\\
&= \norm{\Pi(\delta|\mathcal V)}^2,
\end{align}
where the last step uses Ref.~\cite[Proposition~3B in
Appendix~A.2]{bickel93}.
$\norm{\Pi(\delta|\mathcal V)}^2 = \norm{\delta}^2
-\norm{\Pi(\delta|\mathcal V^\perp)}^2$ follows again from the
Pythagorean theorem for a
$\delta \in \mathcal Z = \mathcal V \oplus \mathcal V^\perp$.
Equation~(\ref{GHB_sub}) comes from Theorem~\ref{thm_GHB}.

\section{\label{sec_classicalT}Proof of Proposition~\ref{classicalT}}
Let $P$ be the true density. For real functions on $\mathbb C$,
define an inner product and a norm
with respect to $P$ as
\begin{align}
\Avg{f,g}_P &\equiv \int d^2\alpha P(\alpha) f(\alpha)g(\alpha),
&
\norm{f}_P &\equiv \sqrt{\Avg{f,f}_P}.
\end{align}
Define the Hilbert space of zero-mean functions as
\begin{align}
\mathcal Z_P &\equiv \BK{f: \norm{f}_P < \infty, \Avg{f,1}_P = 0}.
\end{align}
For each $f \in \mathcal Z_P$, construct the parametric submodel
\begin{align}
\sigma(\theta) &= \int d^2\alpha P(\alpha|\theta)
\ket{\alpha}\bra{\alpha},
\\
P(\alpha|\theta) &= 
\frac{\{1+\tanh[f(\alpha)\theta]\}P(\alpha)}
{\int d^2\alpha \{1+\tanh[f(\alpha)\theta]\}P(\alpha)},
\end{align}
with the truth at $\sigma(0) = \rho$ and
$P(\alpha|0) = P(\alpha)$.  $f(\alpha)$ is the score function with
respect to $P(\alpha|\theta)$. The score with respect to $\sigma$
is then given by
\begin{align}
\rho \circ S &= \rho \circ \bk{\mathcal E f} = 
\int d^2\alpha P(\alpha)f(\alpha) \ket{\alpha}\bra{\alpha},
\end{align}
where the map $\mathcal E:\mathcal Z_P \to \mathcal Z$ is a
quantum version of the conditional expectation \cite{hayashi}.  Hence
\begin{align}
  \BK{\mathcal Ef: f \in \mathcal Z_P} \subseteq \BK{S}.
\end{align}

Consider the inner product between $\mathcal Ef$ and an
$h \in \mathcal B \subset \mathcal Z$ given by
\begin{align}
\Avg{\mathcal Ef,h} &= \trace \rho \Bk{\bk{\mathcal Ef} \circ h} 
= \trace\Bk{\rho \circ \bk{\mathcal Ef}} h
= \Avg{f,\mathcal E^*h}_P,
\end{align}
where Eq.~(\ref{cyclic}) is used and $\mathcal E^*$ is the adjoint map
given by the Husimi representation
\begin{align}
(\mathcal E^*h)(\alpha) &= \bra{\alpha}h\ket{\alpha}.
\end{align}
Since $h \in \mathcal Z$,
\begin{align}
\trace \rho h &= \int d^2\alpha P(\alpha) \bra{\alpha}h\ket{\alpha} = 
\Avg{\mathcal E^*h,1}_P = 0,
\end{align}
and $\mathcal E^*h \in \mathcal Z_P$. The map
$\mathcal E^*:\mathcal B \to \mathcal Z_P$ is obviously linear. It is also
bounded because
\begin{align}
\norm{\mathcal E^*h}_P^2 &= 
\int d^2\alpha P(\alpha) \bk{\bra{\alpha}h\ket{\alpha}}^2
\\
&\le \int d^2\alpha P(\alpha) \bra{\alpha}h^2\ket{\alpha}
= \norm{h}^2.
\end{align}
Thus $\mathcal E^*$ is a continuous linear map
\cite[Theorem~1.5.7]{debnath05}.  As $\mathcal B$ is a dense subset of
$\mathcal Z$ by virtue of Lemma~\ref{cauchy}, $\mathcal E^*$ can be uniquely
extended to a continuous linear map on the whole $\mathcal Z$
\cite[Theorem~1.5.10]{debnath05}.

Any $h \in \mathcal T^\perp$ must obey
\begin{align}
\Avg{\mathcal Ef,h} &= \Avg{f,\mathcal E^*h}_P = 0 \quad \forall f \in \mathcal Z_P.
\end{align}
The only solution is $\mathcal E^*h = 0$. In other words,
$\mathcal T^\perp$ is in the null space of $\mathcal E^*$. As the
Husimi representation is injective \cite{jordan64,*mehta65}, the only
solution to $\mathcal E^*h = 0$ is $h = 0$. Hence
$\mathcal T^\perp = \{0\}$, and $\mathcal T = \mathcal Z$.

\section{\label{sec_Tperp}$\mathcal T^\perp$ for diffraction-limited
  incoherent imaging is infinite-dimensional}
Following Appendix~\ref{sec_classicalT}, it can be shown that
$h \in \mathcal T^\perp$ if
\begin{align}
\bra{\psi_X}h\ket{\psi_X} &= 0
\quad \forall X \in \supp F
\label{Tperp_cond}
\end{align}
for the incoherent-imaging problem in
Sec.~\ref{sec_imaging}. Consider, for example,
$h = \int dk \tilde h(k)\ket{k}\bra{k}$, where $\ket{k}$ is a momentum
eigenket. Then Eq.~(\ref{Tperp_cond}) is satisfied if
\begin{align}
  \bra{\psi_X}h\ket{\psi_X}  &\propto \int dk \tilde h(k) \exp(-2k^2) = 0.
\label{Tperp_cond2}
\end{align}
Let $\{\tilde a_j(k): j \in \mathbb N_0\}$ be the set of Hermite
polynomials that are orthogonal with respect to the weight function
$\exp(-2k^2)$.  Then any $\tilde a_j(k)$ with $j > 0$ satisfies
Eq.~(\ref{Tperp_cond2}).  Define the set
\begin{align}
\BK{a} &\equiv
\BK{a_j = \int dk \tilde a_j(k)\ket{k}\bra{k}: j \in \mathbb N_1}.
\end{align}
Each $a_j$ obeys Eq.~(\ref{Tperp_cond2}) and
\begin{align}
\Avg{a_j,a_k} &\propto \int dk \tilde a_j(k) \tilde a_k(k) \exp(-2k^2)
\propto \delta_{jk},
\end{align}
so $\{a\}$ is an orthogonal set with respect to the inner product given by
Eq.~(\ref{inner}). As $\cspn\{a\} \subseteq \mathcal T^\perp$,
\begin{align}
\dim \mathcal T^\perp &\ge |\BK{a}| = |\mathbb N_1|,
\end{align}
which means that the dimension of $\mathcal T^\perp$ must be infinite.

\section{\label{sec_EC}Derivation of Eq.~(\ref{ECbound})}
For the density-operator family given by Eq.~(\ref{sigma_sub}), the
extended convexity of the Helstrom information \cite{alipour,ng16}
implies that
\begin{align}
K  &\le \tilde K = K^g + K^h,
\\
K_{jk}^g &= 
\left.
\int dX F(X)\parti{g(X|\theta_g)}{\theta_{j}}\parti{g(X|\theta_g)}{\theta_{k}}
\right|_{\theta=0},
\\
K^h_{jk} &= 4\Avg{\Delta k^2}\left.\int dX F(X)
\parti{h(X|\theta_h)}{\theta_{j}}\parti{h(X|\theta_h)}{\theta_{k}}
\right|_{\theta=0},
\end{align}
where $\avg{\Delta k^2} = \bra{\psi_0} k^2\ket{\psi_0}- (\bra{\psi_0}
k\ket{\psi_0})^2 = 1/4$ is the variance of $k$. With the explicit
partition of $\theta$ into $\theta_g$ and $\theta_h$, $\tilde K$
can be expressed as
\begin{align}
\tilde K &= \begin{pmatrix}K^g & 0\\ 0 & K^h\end{pmatrix}
\\
&= \int dX F(X)
\begin{pmatrix}(\partial_g g)(\partial_g g)^\top & 0\\ 0 & 
(\partial_h h)(\partial_h h)^\top \end{pmatrix},
\end{align}
where $\partial_g = (\partial_{g1},\partial_{g2},\dots)^\top$ and
$\partial_h = (\partial_{h0},\partial_{h1},\dots)^\top$. Let
\begin{align}
g(X|\theta_g) &= \sum_{j=1}^\infty \theta_{gj} a_j(X),
\\
h(X|\theta_h) &= X + \sum_{j=0}^\infty \theta_{hj} a_j(X),
\end{align}
where $\{a_j(X): j \in \mathbb N_0\}$ is a set of orthogonal
polynomials with respect to the true $F$ that satisfy
$\int dX F(X) a_j(X) a_k(X) = \delta_{jk}$.  $a_0(X) = 1$ is omitted
from $g(X|\theta_g)$ because $g$ is a score function with respect to
$F$ and $\int dX F(X) g(X|\theta_g) = 0$ implies that $g(X|\theta_g)$
cannot contain $a_0(X)$ in its expansion. The orthonormality of
$\{a\}$ leads to
\begin{align}
\tilde K &= I, & \tilde K^{-1} &= I.
\end{align}
Now consider
\begin{align}
\beta_\mu(\theta) &= \int dX F(X|\theta) X^\mu
\\
&= \int dY G(Y|\theta_g) \Bk{h(Y|\theta_h)}^\mu,
\\
\left.\parti{\beta_\mu}{\theta_{gj}}\right|_{\theta = 0} &= 
\int dX F(X) a_j(X) X^\mu,
\\
\left.\parti{\beta_\mu}{\theta_{hj}}\right|_{\theta = 0} &= 
\mu \int dX F(X)    X^{\mu-1} a_j(X).
\end{align}
Then 
\begin{align}
\bk{\partial\beta_\mu}^\top \tilde K^{-1}\partial\beta_\mu
&= 
\sum_{j=1}^\infty \bk{\parti{\beta_\mu}{\theta_{gj}}}^2 + 
\sum_{j=0}^\infty \bk{\parti{\beta_\mu}{\theta_{hj}}}^2
\\
&= \beta_{2\mu} - \beta_\mu^2 + \mu^2 \beta_{2\mu-2},
\end{align}
where the completeness property
\begin{align}
\sum_{j=0}^\infty a_j(X) a_j(X') = 
1 + \sum_{j=1}^\infty a_j(X) a_j(X') = \delta(X-X')
\end{align}
is assumed. With
\begin{align}
\mathsf{H}_\mu^\sigma &= \bk{\partial\beta_\mu}^\top K^{-1}\partial\beta_\mu
\ge \bk{\partial\beta_\mu}^\top \tilde K^{-1}\partial\beta_\mu,
\end{align}
and using Corollary~\ref{iid2} and Proposition~\ref{loose},
Eq.~(\ref{ECbound}) is obtained.

\section{\label{sec_Seff}Proof of Theorem~\ref{thm_Seff}}
The proof follows the classical case \cite{tsiatis06}.
As $S^\beta \notin \Lambda$, the $S_{\rm eff}$ given by Eq.~(\ref{Seff})
is not zero. Let 
\begin{align}
\delta &= \frac{S_{\rm eff}}{\norm{S_{\rm eff}}^2}.
\label{delta_Seff}
\end{align}
Notice that Eq.~(\ref{Seff}) is a projection of $S^\beta$ into
a space orthogonal to $\Lambda$, so $S_{\rm eff} \perp \Lambda$
and $\delta \perp\Lambda$. Then
\begin{align}
\Avg{S^\beta,\delta} &= \frac{1}{\norm{S_{\rm eff}}^2}
\Avg{S_{\rm eff}+\Pi(S^\beta|\Lambda),S_{\rm eff}} = 1,
\\
\Avg{S_j^\eta,\delta} &= 0,
\end{align}
because $\Pi(S^\beta|\Lambda) \in \Lambda$ and each
$S_j^\eta \in \Lambda$.  Thus $\delta$ satisfies
Eqs.~(\ref{unbiased_eta}) and is an influence operator.  Notice also
that $S_{\rm eff}$ and $\delta$ are in $\mathcal T$, because
$S^\beta \in \mathcal T$ and
$\Pi(S^\beta|\Lambda) \in \Lambda \subseteq \mathcal T$.  Hence, by
Theorem~\ref{thm_SHB},
\begin{align}
\delta_{\rm eff} &= \Pi(\delta|\mathcal T) = \delta,
\end{align}
and Eq.~(\ref{delta_Seff}) is the efficient influence.

\section{\label{sec_GHB_vec}Proof of Theorem~\ref{thm_GHB_vec}}
We again follow Ref.~\cite{tsiatis06}. Decompose any
$\delta \in \mathcal D \subseteq \mathcal Z = \mathcal T \oplus\mathcal
T^\perp$ into
\begin{align}
\delta &= \delta_{\rm eff} + h,
&
\delta_{\rm eff} &= \Pi(\delta|\mathcal T),
&
h &= \Pi(\delta|\mathcal T^\perp).
\end{align}
It is straightforward to prove that $\delta_{\rm eff} \in \mathcal D$.
As $h$ is orthogonal to any element in
$\mathcal T \equiv (\cspn\{S\})^{\oplus q}$, it must be orthogonal to
$g = (0,\dots,0,e,0,\dots,0)^\top$ with any $e \in \cspn\{S\}$ in any
entry of $g$, say, the $j$th entry. Then
\begin{align}
\trace\Avg{h,g} &= \Avg{h_j,e} = 0,
\end{align}
meaning that each entry of $h$ is orthogonal to $\cspn\{S\}$.  This
leads to a stronger matrix form of the orthogonality between
$\delta_{\rm eff} \in \mathcal T$ and $h \in \mathcal T^\perp$ given
by
\begin{align}
\Avg{\delta_{\rm eff},h} &= 0,
\end{align}
and a matrix form of the Pythagorean theorem given by
\begin{align}
\Avg{\delta,\delta} &= \Avg{\delta_{\rm eff},\delta_{\rm eff}}
+ \Avg{h,h} \ge \Avg{\delta_{\rm eff},\delta_{\rm eff}},
\end{align}
resulting in Eq.~(\ref{GHB_vec}). To prove the uniqueness of
$\delta_{\rm eff}$ in $\mathcal D$, suppose that there exists another
$\delta' \in \mathcal D$ that gives
$\avg{\delta',\delta'} = \avg{\delta_{\rm eff},\delta_{\rm eff}}$.
Define $g = \delta'-\delta_{\rm eff}$. As
$\avg{S,g} = \avg{S,\delta'} - \avg{S,\delta_{\rm eff}} =
\partial\beta-\partial\beta = 0$, $g \in \mathcal T^\perp$, and the
matrix Pythagorean theorem gives
$\avg{\delta',\delta'} = \avg{\delta_{\rm eff},\delta_{\rm eff}} +
\avg{g,g}$. This implies that $\avg{g,g} = 0$,
$\norm{g}^2 = \trace \avg{g,g} = 0$, and $g = 0$, contradicting the
assumption that $\delta' \neq \delta_{\rm eff}$. Hence
$\delta_{\rm eff}$ must be unique.

\section{\label{sec_belavkin}Proof of Lemma~\ref{lem_belavkin}}

  Let the superscript $*$ denote the entry-wise conjugation of a
  matrix and the superscript $\dagger = *\top$ denote the conjugate
  transpose.  $A \ge 0$ means that $z^\dagger A z\ge 0$ for any
  $z \in \mathbb C^q$. We also have $A^* \ge 0$, since
  $z^\dagger A^* z = (z^{*\dagger} A z^*)^* = z^{*\dagger} A z^* \ge
  0$ for any $z \in \mathbb C^q$. Thus, for any $z \in \mathbb C^q$,
\begin{align}
z^\dagger\bk{\real A \pm i \imag A} z &\ge 0,
&
z^\dagger (\real A)z &\ge \abs{z^\dagger (i\imag A) z}.
\end{align}
Let $\{\lambda_s,z_s: s = 1,\dots,q\}$ be the eigenvalues and
eigenvectors of the Hermitian $i\imag A$.  As the singular values of
$i\imag A$ are $\{|\lambda_s|\}$, we obtain
\begin{align}
\trace\real A &= \sum_s z_s^\dagger(\real A) z_s \ge \sum_s
\abs{z_s^\dagger(i\imag A)z_s}
\nonumber\\
&= \sum_s |\lambda_s| = \norm{i\imag A}_1 =
\norm{\imag A}_1.
\end{align}


\section{\label{sec_parametric}Some results
concerning quantum multiparameter estimation}
This appendix presents some interesting results concerning quantum
multiparameter estimation, following Sec.~\ref{sec_multi} and assuming
$1\le q \le p <\infty$.

A crucial assumption in this paper is that $\mathcal D$, the set of
influence operators, is not empty. While this is not a problem for all
the examples studied in this paper, the following theorem,
generalizing a classical result by Stoica and Marzetta
\cite{stoica01}, can be used to verify the assumption.

\begin{theorem}
  $\mathcal D$ is not empty if and only if all the columns of
  $\partial\beta$ are in the range of the Helstrom information matrix
  $K$, viz.,
\begin{align}
KK^+\partial\beta &= \partial\beta,
\label{range}
\end{align}
where the superscript $+$ denotes the Moore-Penrose pseudoinverse
\cite{golub13}.

\label{thm_range}
\end{theorem}

\begin{proof}
  We prove the ``only if'' part first. Assume that a
  $\delta \in \mathcal D$ exists.  It satisfies
  $\avg{S,\delta} = \partial\beta$, and therefore
\begin{align}
u^\top \Avg{S,\delta} v = \Avg{u^\top S,v^\top\delta} = u^\top\bk{\partial\beta}v,
\end{align}
for any $u \in \mathbb R^p$ and $v \in \mathbb R^q$. The Cauchy-Schwartz
inequality gives
\begin{align}
\abs{u^\top\bk{\partial\beta}v}^2 &\le 
\bk{u^\top K u} \bk{v^\top\Avg{\delta,\delta}v}.
\end{align}
Now suppose that $u$ is in the null space of $K$, such that $Ku = 0$,
and pick $v = (\partial\beta)^\top u$. We obtain
\begin{align}
\abs{u^\top\bk{\partial\beta}(\partial\beta)^\top u}^2 \le 0,
\end{align}
which implies $(\partial\beta)^\top u = 0$. As this must hold for any
$u$ in the null space of $K$, each column of $\partial\beta$ must be
orthogonal to the null space and therefore in the range of $K$.
$KK^+$ is the projection matrix into the range of $K$ \cite{golub13}, so
Eq.~(\ref{range}) holds.

The ``if'' part comes from the fact that,
as long as Eq.~(\ref{range}) holds,
\begin{align}
\delta &= (\partial\beta)^\top K^+ S
\label{an_influence}
\end{align}
satisfies $\avg{\delta,I} = 0$ and
$\avg{S,\delta} = K K^+ \partial\beta = \partial\beta$ and is
therefore an influence operator.
\end{proof}
For an illustrative example, consider
\begin{align}
\theta &= \begin{pmatrix}\theta_1\\ \theta_2\end{pmatrix},
&
\partial\beta &= 
\begin{pmatrix}a \\ b\end{pmatrix},
\end{align}
with the geometry depicted in Figure~\ref{singular}.  $S_1 = 0$ and
$K_{11} = \avg{S_1,S_1} = 0$ at the singular point $\theta = \varphi$,
meaning that
\begin{align}
K(\varphi) &= \begin{pmatrix}0 & 0\\ 0 & \Avg{S_2,S_2}\end{pmatrix}.
\end{align}
The tangent space there becomes a line in the $S_2$ direction, and it
is impossible for a $\delta$ to satisfy
\begin{align}
\Avg{S,\delta} &= \begin{pmatrix} 0 \\ \Avg{S_2,\delta} \end{pmatrix}
= \begin{pmatrix}a \\ b\end{pmatrix},
\end{align}
if $a \neq 0$.

\begin{figure}[htbp!]
\centerline{\includegraphics[width=0.45\textwidth]{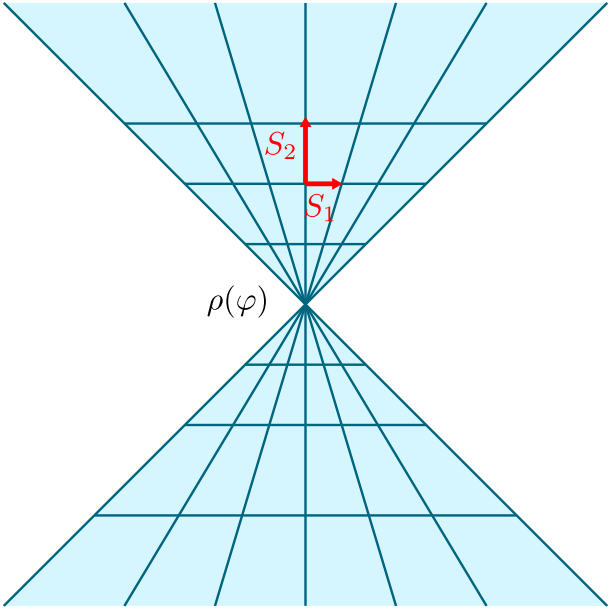}}
\caption{\label{singular} 
The lines denote the Cartesian grid of
  $\theta_1$ and $\theta_2$.  The distortion of the grid represents
  the geometry of $\rho(\theta_1,\theta_2)$.  For example, the
  distance from one grid point $(\theta_1,\theta_2)$ to a neighboring
  $(\theta_1+\Delta,\theta_2)$ in the figure represents the distance
  between $\rho(\theta_1,\theta_2)$ and
  $\rho(\theta_1+\Delta,\theta_2)$. $S_1$ and $S_2$ are tangent
  vectors, and $K = \avg{S,S}$ is a metric that governs the
  infinitesimal distance between two neighboring $\rho$'s.  At the
  singular point $\rho(\varphi)$, $K_{11} = 0$, the tangent space
  becomes a line in the $S_2$ direction, which forbids the existence
  of influence operators for certain $\partial\beta$.}
\end{figure}

If Eq.~(\ref{range}) does not hold at certain values of $\theta$,
Theorem~\ref{thm_range} implies that an unbiased estimator of $\beta$
cannot exist there, and the GHB can be assumed to be infinite.  Note,
however, that a biased estimator may still be able to achieve a finite
error.

Provided that Eq.~(\ref{range}) holds, a pseudoinverse form of the
Helstrom bound can be obtained.
\begin{corollary}
\label{cor_pseudoH}
If Eq.~(\ref{range}) holds, 
\begin{align}
\delta_{\rm eff} &= (\partial\beta)^\top K^+ S,
\\
\Avg{\delta_{\rm eff},\delta_{\rm eff}}
&= (\partial\beta)^\top K^+ \partial\beta,
\label{efficient_covar}
\\
\tilde{\mathsf H} &= 
\trace W (\partial\beta)^\top K^+ \partial\beta.
\label{pseudoH}
\end{align}
\end{corollary}
\begin{proof}
  Equation~(\ref{an_influence}) is an influence operator and also a
  linear combination of $S$, so it is in the tangent space
  $\mathcal T$.  By Theorem~\ref{thm_GHB}, it must be efficient. The
  other results follow from the fact $K^+KK^+ = K^+$ \cite{golub13}
  and the definition of $\tilde{\mathsf H}$.
\end{proof}

The original Helstrom bound is a simple consequence, generalizing the
scalar version in Corollary~\ref{cor_HB}.
\begin{corollary}
  If $K > 0$, 
\begin{align}
\delta_{\rm eff} &= (\partial\beta)^\top K^{-1} S,
\\
\Avg{\delta_{\rm eff},\delta_{\rm eff}}
&= (\partial\beta)^\top K^{-1} \partial\beta,
\\
\tilde{\mathsf H} &=  
\trace W (\partial\beta)^\top K^{-1} \partial\beta
\equiv \mathsf H.
\end{align}
\label{cor_helstrom_vec}
\end{corollary}
\begin{proof}
  If $K > 0$, $K^{-1}$ exists, $K^+ = K^{-1}$, Eq.~(\ref{range})
  always holds, and the results follow from
  Corollary~\ref{cor_pseudoH}.
\end{proof}

Finally, we mention that the semidefinite program presented
  in~Ref.~\cite{albarelli19} to evaluate the Holevo bound for
  $\beta=\theta$ and a nonsingular $K$ can be straightforwardly
  extended to the more general setup considered in this appendix.


  \section{Post-publication notes}

  After the completion of this work and its acceptance for
  publication, Masahito Hayashi informed us that
  Ref.~\cite[(c)]{yang19} and Refs.~\cite{suzuki19,suzuki19b,suzuki20}
  also study quantum estimation theory with nuisance parameters.  In
  particular, Ref.~\cite[Sec.~4.3]{suzuki20} independently arrives at
  results similar to our Theorem~\ref{thm_holevo} and
  Corollary~\ref{cor_helstrom_vec}. Reference~\cite{suzuki20} focuses
  on the parametric case ($p < \infty$), whereas our
  Theorems~\ref{thm_GHB_vec} and \ref{thm_holevo} are proven to work
  in both parametric and semiparametric settings.

  We note that our Theorems~\ref{thm_GHB_vec}, \ref{thm_holevo},
  \ref{thm_sandwich}, \ref{thm_range} and
  Corollaries~\ref{cor_holevo_scalar}--\ref{cor_helstrom_vec} first
  appear in an arXiv preprint of ours on February 5th, 2020
  \cite[v2]{albarelli20}. We then decided to merge our two preprints
  \cite{tsang20,albarelli20} into one manuscript \cite[v6]{tsang20},
  which was accepted by PRX on June 1st, 2020. On the other hand, the
  first appearance of Sec.~4.3 in Ref.~\cite{suzuki20} seems to be in
  the Accepted Manuscript on the JPA website on April 21st, 2020---the
  section is absent in v1 and v2 of their arXiv preprint
  \cite{suzuki19}.

  On another note, Ref.~\cite[Remark~4.5 in v3]{suzuki19} proves that,
  if $\dim \mathcal T^\perp < \infty$ and $W > 0$, then there exists a
  minimizing solution in $\mathcal D$ for the Holevo bound given by
  Eq.~(\ref{X}).


\bibliography{research2}

\end{document}